\documentclass[a4paper,11pt]{article}
\usepackage{graphicx} 

\usepackage[utf8]{inputenc}
\usepackage[T1]{fontenc}
\usepackage{titlesec}

\usepackage{amsmath,amsfonts,amssymb,amsthm}
\usepackage{fullpage}
\usepackage{subfigure}
\usepackage{epstopdf}
\usepackage{color}
\usepackage[sans]{dsfont}
\usepackage{mathtools}
\usepackage[colorlinks=true,allcolors=blue!80!black, linktocpage]{hyperref}

\usepackage{appendix}

\usepackage{cleveref}
\crefname{algocf}{algorithm}{algorithms}

\usepackage{tcolorbox}
\usepackage{multirow}
\usepackage[shortlabels]{enumitem}
\setlist[enumerate]{label=(\arabic*)}
\setlist[itemize]{noitemsep}

\usepackage{todonotes}
\usepackage{xspace}
\tcbuselibrary{skins,breakable}
\tcbset{enhanced jigsaw}
\usepackage{microtype}
\usepackage{newtxtext}
\usepackage{newtxmath}

\usepackage[backend=biber,style=alphabetic,
 maxnames=10,
  url=false,     
  isbn=false     
            ]{biblatex}
\usepackage[margin=2cm,right=3cm]{geometry}

\usepackage[linesnumbered, ruled]{algorithm2e}
\SetKwRepeat{Do}{do}{while}%
\usepackage{multirow}
\usepackage{makecell}
\usepackage{algpseudocode}

\usepackage{comment}
\usepackage{xcolor}
\usepackage{nameref}
\usepackage{array}
\usepackage{tikz}
\usetikzlibrary{decorations.pathmorphing}
\usepackage{marginnote}
\usepackage[normalem]{ulem}

\usepackage{thmtools}

\newtheorem{theorem}{Theorem}
\newtheorem{corollary}[theorem]{Corollary}
\newtheorem{lemma}[theorem]{Lemma}

\theoremstyle{definition}

\newtheorem{definition}[lemma]{Definition}

\usepackage{thm-restate}

\newcommand{\papertitle}{Faster Directed Hopsets, Distance Preservers, and Deterministic Shortcut Sets}

\newif\ifanonymous
\ifx\authoranonymous\relax\anonymoustrue\fi

\makeatletter
\newcommand{\pushright}[1]{\ifmeasuring@#1\else\omit\hfill$\displaystyle#1$\fi\ignorespaces}
\newcommand{\pushleft}[1]{\ifmeasuring@#1\else\omit$\displaystyle#1$\hfill\fi\ignorespaces}
\makeatother

\newcommand{\set}[1]{\left\{ #1 \right\}}

\newcommand{\qset}{{\mathcal{Q}}}

\newcommand{\lset}{{\mathcal{L}}}

\newcommand{\rset}{{\mathcal{R}}}

\newcommand{\polylog}{\text{\normalfont polylog }}

\newcommand{\eps}{{\varepsilon}}
\newcommand{\smol}{{\eps}}
\newcommand{\poly}{\mathsf{poly}}

\newcommand{\tc}{\mathrm{TC}}

\newcommand{\dist}{\textnormal{\textsf{dist}}}

\newcommand{\chains}{\mathcal C}

\newcommand{\cO}{\mathcal O}
\newcommand{\BigO}{\cO}
\newcommand{\Oish}{\widetilde{\cO}}
\newcommand{\setsize}[1]{\left|#1\right|}

\DeclareMathOperator{\supershortcut}{SuperShort}

\DeclareMathOperator{\coverCount}{\varphi}
\newcommand{\opCount}{k}

\newcommand{\treefinder}{\mathtt{TreeFinder}}

\newcommand{\pre}{\mathrm{pre}}
\newcommand{\suf}{\mathrm{suf}}

\newlength{\LETTERheight}
\AtBeginDocument{\settoheight{\LETTERheight}{I}}
\newcommand*{\longleadsto}[1]{\ \raisebox{0.24\LETTERheight}{\tikz \draw [->,
line join=round,
decorate, decoration={
    zigzag,
    segment length=4,
    amplitude=.9,
    post=lineto,
    post length=2pt
}] (0,0) -- (#1,0);}\ }

\crefname{enumi}{condition}{conditions}
\crefname{algocf}{algorithm}{algorithms}

\def\ceil#1{\left\lceil #1 \right\rceil}

\newif\ifhidetodos
\hidetodostrue

\newcommand{\cbox}[2][yellow]{%
  \fcolorbox{#1}{white}{\parbox{\dimexpr\linewidth-2\fboxsep}{\strut #2\strut}}%
}

\newcommand{\customtodo}[3]{\textcolor{#2}{\(\blacktriangledown\)}\marginnote{\raggedright \textcolor{#2}{\textbf{#1:} #3}}}
\newcommand{\customtododisplay}[3]{\noindent\cbox[#2]{\textcolor{#2}{\textbf{#1:} #3}}}
\newcommand{\customtodoinline}[3]{\textcolor{#2}{\textcolor{#2}{\textbf{#1:} #3}}}

\ifhidetodos
    \renewcommand{\customtodo}[3]{}
    \renewcommand{\customtododisplay}[3]{}
    \renewcommand{\customtodoinline}[3]{}
\fi

\newcommand{\yn}[1]{\customtodo{YN}{purple}{#1}}

\makeatletter
\g@addto@macro\appendix{\crefalias{section}{appendix}\crefalias{subsection}{subappendix}}
\makeatother

\makeatletter
\@ifpackageloaded{enumitem}{%
  \newlist{conditions}{enumerate}{1}%
  \setlist[conditions]{nosep,leftmargin=2.2em,label=(A\arabic*{})}%
  \crefname{conditionsi}{condition}{conditions}%
}{%
  \newenvironment{conditions}{\begin{enumerate}[({A}1)]}{\end{enumerate}}%
}
\makeatother

\setlist[enumerate]{label=(\arabic*),nosep}
\setlist[1]{labelindent=\parindent}

\title{\papertitle}
\date{April 24, 2026}

\usepackage{authblk}
\usepackage{orcidlink}

\author[1,2]{Ben Bals \orcidlink{0009-0009-1054-8444}}
\author[1]{Daniel Dadush \orcidlink{0000-0001-5577-5012}}
\author[3]{Gary Hoppenworth \orcidlink{0000-0002-6534-8935}}
\author[1,2]{\protect\\ Yasamin Nazari \orcidlink{0000-0003-1315-9355}}
\author[2]{Rajath Rao K.N \orcidlink{0000-0003-3422-6940}}
\affil[1]{CWI, Amsterdam, The Netherlands}
\affil[2]{Vrije Universiteit Amsterdam, The Netherlands}
\affil[3]{TTIC, Chicago, United States}

\begin{document}

\maketitle

\begin{abstract}
   While there have been many recent advances in getting better tradeoffs for directed distance structures such as hopsets, shortcut sets and distance preservers, most known algorithms are inefficient. In this work we provide three results:
    \begin{itemize}
        \item Faster algorithms for directed $(1+\epsilon)$-hopsets, matching the state-of-the-art size/hopbound trade-offs of \cite{bernstein_closing_2023}. Our algorithm is first designed for DAGs and  then extended to general graphs (up to $n^{o(1)}$ factors) using the recent DAG projection result of \cite{haeuplerJiangSaranurak26}.
        \item Faster \textit{deterministic} algorithms for constructing shortcut sets, matching the state-of-the-art size/hopbound tradeoffs of \cite{KP22a}.
        \item Using our improved directed hopset construction, we get significantly faster algorithms for computing source-wise distance preservers. This algorithm is based on speeding-up reductions of \cite{KP22} by using certain structural properties of our hopsets.
    \end{itemize}
   
\end{abstract}

\clearpage
\tableofcontents
\clearpage

\section{Introduction}
Shortcut sets (Thorup \cite{thorup92}) and hopsets (Cohen \cite{cohen00}) are fundamental objects that have been studied explicitly or implicitly for over three decades. They received significant attention due to their vast applications in reachability and shortest path computation in various models (e.g.~\cite{UY91, KS97, HKN14a, HKN14b, HKN15, forsterNanongkai18, JLS19, fineman2018, GW20, BGW20, CFR20, KS21, ASZ20}), in addition to their extremal properties and their connections with related objects such as reachability and distance preservers (\cite{KP22,koganParter25missing}). 

\begin{definition}
Given a directed graph $G= (V,E)$, a \textbf{shortcut set} is a set of additional (unweighted) edges $H$ in the transitive closure\footnote{Recall: The transitive closure of $G$ is the graph over the same vertex set as $G$, which has an edge $(u, v)$ iff there is a $u \leadsto v$ path in $G$.} of $G$.
The \textbf{hopbound} $\beta$ of $H$ is the maximum hop distance in $G \cup H$ over all reachable pairs.

For a weighted graph $G$, a $(1+\epsilon)$-\textbf{hopset} with hopbound $\beta$ is a set of additional weighted edges $H$ such that for each
pair $(u,v)$ in the transitive closure, we have $d_G(u,v) \leq d^{(\beta)}_{G \cup H} (u,v) \leq (1+\varepsilon)d_G(u, v)$, where $d^{(\beta)}_{G \cup H} (u,v)$ denotes the length of a shortest path in $G \cup H$ using at most $\beta$ hops (edges).
\end{definition}

Hopsets are very well-understood in undirected graphs where almost matching upper \cite{elkin2019journal} and lower \cite{abboudBodwinPettie18} bounds and fast constructions are known. However, in directed graphs both existentially and computationally the picture is much less complete. 
For a long time, the only known approach for directed graphs was the following folklore algorithm: randomly sample a set of $\Oish(n^{1/2})$ vertices,\footnote{We use $\Oish_{\varepsilon}(\cdot)$ to hide $\polylog(n)$ and $\poly(1/\varepsilon)$ factors.} and add a shortcut edge
between each pair of sampled nodes that lie in the transitive closure of the input graph. For hopsets we additionally give a weight corresponding to the distance in $G$ between the endpoints of each added edge.
This simple approach leads to a shortcut set (or an exact hopset) of $\Oish(n)$ size and hopbound $\Oish(n^{1/2})$. A line of work has focused on optimizing the running time of a similar size/hopbound tradeoff with a focus on applications in parallel reachability and shortest path computation \cite{fineman2018,JLS19,CFR20}. These algorithms are also sampling based but they sample a sequence of pivots in a careful recursive approach that leads to $m^{1+o(1)}$ running time for shortcut sets and $(1+\epsilon)$-hopsets with near-linear size and $n^{1/2+o(1)}$-hopbound.

Remarkably, the hopbound of $\Oish(n^{1/2})$ for linear size shortcut sets achieved by the folklore sampling remained nearly unimproved for 30 years, until a breakthrough of Kogan and Parter \cite{KP22a} that proved: every graph has a shortcut set on $O(n)$ edges with hopbound $\Oish(n^{1/3})$. They also gave a weaker tradeoff (still improving over the folklore algorithm) for hopsets. Later, Bernstein and Wein \cite{bernstein_closing_2023} matched the size/hopbound shortcut set tradeoff of \cite{KP22a} also for $(1+\epsilon)$-hopsets. The hopset algorithm of \cite{bernstein_closing_2023} is much more complicated, but this complexity was mainly needed for handling general graphs. For DAGs, we can show a simpler hopset algorithm (see \Cref{sec:kp-hopset} for details), which was implicit in \cite{KP22a}; as observed in \cite{haeuplerJiangSaranurak26}, combining it with DAG projections gives a simpler construction for general graphs up to $n^{o(1)}$ factors.
At the same time a sequence of lower bounds \cite{hesse03,huangPettie21, bodwinHoppenworth23,hoppenworthXuXu25} were shown. The current best lower bound for $O(n)$-size shortcut sets is $\widetilde{\Omega}(n^{1/4})$. Closing the upper and lower bound gap for shortcut sets is a major open problem.

\subparagraph{Faster Hopsets.} Despite the progress on the existential bounds, these constructions are still computationally inefficient.
Initially, the shortcut sets of \cite{KP22a} and hopsets of \cite{bernstein_closing_2023} required transitive closure and all-pairs-shortest paths time respectively. While these objects are interesting in their own right, such running times would rule out most conceivable applications in reachability and shortest path computations.
Motivated by this bottleneck, \cite{kogan_faster_nodate} showed how shortcut sets can be constructed in faster than transitive closure time, using a reduction to min cost flow for computing chain covers.
(Chains are paths in the transitive closure and a chain cover aims to cover as many vertices as possible with a fixed number of chains.)
They thus show how to return paths on the transitive closure without directly computing the transitive closure.
However so far no faster algorithms were known for the $(1+\epsilon)$-hopset constructions with improved size/hopbound tradeoffs over the folklore approach. The main challenge is that for hopsets we need these chain covers to also correspond to shortest paths between their endpoints, and we do not know how to enforce such a requirement in the min cost flow reduction of \cite{kogan_faster_nodate}.
Using a greedy algorithm instead, we get the following running time for constructing the hopsets whose size/hopbound trade-off matches the  state-of-the-art hopsets of \cite{bernstein_closing_2023}:

\begin{restatable}{theorem}{hsalldiamthm}\label{thm:hs-all-diam}
Let $G=(V,E,w)$ be an $n$-vertex, $m$-edge directed acyclic graph with nonnegative edge-weights and polynomial aspect ratio. For any $\varepsilon \in (0, 1)$ and hopbound $\beta$, there is a randomized
$\Oish_{\varepsilon}\!\left(\frac{n}{\beta}\cdot m \right)$-time algorithm computing
$\Oish_{\varepsilon}\!\left( \max\left( \frac{n^2}{\beta^3}, (\frac{n}{\beta})^{\frac32} \right)\right)$-size $(1+\varepsilon)$-hopsets with hopbound $\cO(\beta)$.
\end{restatable}

We note that while our theorem statement assumes the input graph is a directed acyclic graph (DAG), this is without loss of generality (up to $n^{o(1)}$ factors) due to a reduction of \cite{haeuplerJiangSaranurak26} (see  \cref{app:dagreduction} for details). 

For the important case of $\Oish(n)$-size hopsets with hopbound $O(n^{1/3})$, by setting $\beta = n^{1/3}$ we get $\Oish(mn^{2/3})$ running time. This improves the $O(mn)$ running time of \cite{bernstein_closing_2023}. Compare this to the best (randomized) construction time of $\Oish(mn^{1/3})$ for shortcut sets \cite{kogan_faster_nodate}. 


\subparagraph{Deterministic Shortcut Sets.\protect \footnote{We would like to point out that this part of the work was previously included in another paper (and also appeared in an earlier version of an arXiv preprint \cite{balsbodwin2025greedy}). This result was moved to this paper, which we believe is a better technical and conceptual fit for it. }}
A limitation of most known shortcut set algorithms is their reliance on randomness. With the exception of a recent work \cite{balsbodwin2025greedy} that gives a greedy deterministic algorithm for shortcut sets that match tradeoffs of \cite{KP22a}, all previous approaches required sampling vertices, or paths, or both. On the other hand, the deterministic greedy approach of \cite{balsbodwin2025greedy} is much slower ($\Omega(n^2m)$ time), as it is based on an iterative algorithm that requires examining a potential function for adding each edge.

Our second result is a faster deterministic algorithm for shortcut sets that match the existential tradeoffs of \cite{KP22a}. For the linear size shortcut sets of hopbound $\Oish(n^{1/3})$ this also leads to running time of $\Oish(mn^{2/3})$. While the running time does not match the best randomized algorithm \cite{kogan_faster_nodate}, it is still faster than transitive closure time, and is much faster than the best known deterministic algorithm implied by \cite{balsbodwin2025greedy}.

\begin{restatable}[Deterministic Shortcut Sets]{theorem}{thmSetcover}
    \label{thm:det-shortcut}
    For any $\beta = \BigO(n^{1/3})$, there is a deterministic $\Oish(m\cdot n/\beta)$-time
    algorithm computing $\Oish(n^2/\beta^3)$-size $\BigO(\beta)$-shortcut sets. 
\end{restatable}

We also make a simple observation that gives a deterministic almost-linear time algorithm for computing $\BigO(\sqrt{n})$-shortcut sets.

\subparagraph{Faster Distance Preservers.}
Given an input graph $G$ and collection of demand pairs $P \subseteq V(G) \times V(G)$, a distance preserver is a sparse subgraph of $G$ that preserves (or approximates) distances between all demand pairs in $P$. Distance preservers were first introduced and studied by \cite{bollobasCoppersmithElkin03} and  \cite{coppersmithElkin06}.  Since then,  significant progress has been made in understanding the combinatorial bounds on distance preservers~\cite{bodwinWilliams16,bodwin17,bodwin21,cyganGrandoniKavitha13,kavithaVarma15,KP22,bernstein_closing_2023,koganParter25missing,hoppenworthXuXu25,haeuplerJiangSaranurak26}. Also, distance preservers (for undirected graphs) have found applications in computing path-reporting distance oracles~\cite{pettie09,elkinPettie16,ENW16,bodwinGrandoniParterWilliams17,DBLP:conf/focs/ElkinS23}. 

While distance preservers are widely studied from an extremal point of view, to our knowledge there were no known computationally efficient algorithms for directed graphs. The known applications in distance oracles were exclusively applied to undirected graphs, and this lack of efficient constructions could be a reason for the lack of such applications for directed graphs.
For other related objects, the closest is the work of \cite{AB18} achieved fast constructions of \emph{reachability preservers}.

Distance preservers have been studied in many different settings (e.g., directed/undirected graphs, weighted/unweighted graphs, and approximate/exact distances). In particular, we study source-wise approximate distance preservers in directed weighted graphs.

\begin{definition}[$(1+\varepsilon)$-approximate distance preservers]
    Let $G = (V, E, w)$ be an $n$-vertex directed graph, let $S \subseteq V$ be a set of source nodes, and let $P \subseteq S \times V$ be collection of demand pairs. Then for any $\varepsilon > 0$, a sourcewise $(1+\varepsilon)$-approximate distance preserver $H \subseteq G$ is a sparse subgraph such that for all $(s, t) \in P$, $\dist_H(s, t) \le (1+\varepsilon) \dist_G(s, t)$. 
\end{definition}

In directed graphs, $(1+\varepsilon)$-approximate distance preservers were first studied explicitly in \cite{KP22}, which established a blackbox reduction from $(1+\varepsilon)$-approximate hopsets to $(1+\varepsilon)$-approximate distance preservers. This reduction, combined with the subsequent work of \cite{bernstein_closing_2023}, established that every $n$-vertex directed graph $G$ and set of demand pairs $P$ of size $|P| = p$ admits a $(1+\varepsilon)$-approximate distance preserver $H \subseteq G$ of size $|E(H)| = \Oish(n^{2/3}p^{2/3} + np^{1/3})$. This is the state-of-the-art\footnote{When $p \le n^{3/4}$ the recent work of \cite{haeuplerJiangSaranurak26} obtains improved bounds.} bound on $|E(H)|$ when $p \ge n^{3/4}$.
\yn{should add running time and compare BB: See end of subsection}

We then combine our fast $(1+\varepsilon)$-hopsets with the ideas from hopset-to-preserver reduction of \cite{KP22}, to get a faster $(1+\varepsilon)$-approximate distance preservers.
We note that we cannot use this reduction directly for obtaining our desired running time, and need to white-box the reduction and adapt it to the specific properties of the hopset used.

\begin{restatable}[Fast distance preservers]{theorem}{FastSourcewisePreservers}
\label{thm:fast-sourcewise-preservers}
Let $G=(V,E,w)$ be an $n$-vertex DAG with $m$ edges, nonnegative edge-weights,
and polynomial aspect ratio. Let $S\subseteq V$ and
$P\subseteq S\times V$, and write $\sigma:=|S|$ and $p:=|P|$.
Then one can compute a $(1+\epsilon)$-approximate  distance preserver
$G_P\subseteq G$ of $G, P$ with 
$|E(G_P)| = \widetilde O\!\left(np^{1/3} + (np)^{2/3}\right)$ edges in running time
\begin{itemize}
    \item $\widetilde O\!\left(p^{2/3}m + \sigma (m+p) + np^{1/3}\right)$ when $p \le n$, and
    \item $\widetilde O\!\left((np)^{1/3}m + \sigma (m+p) + (np)^{2/3}\right)$ when $p \ge n$.
\end{itemize}
\end{restatable}

In particular, the previous theorem implies the following corollary when our demand pairs take the form $P = S \times V$. 

\begin{corollary}
\label{corr:preserver}
Let $G=(V,E,w)$ be an $n$-vertex DAG with $m$ edges, nonnegative edge-weights,
and polynomial aspect ratio. Let $S\subseteq V$  be a set of $\sigma:=|S|$ source vertices. Then one can compute a $(1+\varepsilon)$-approximate distance preserver $H \subseteq G$ for demand pairs  $S \times V$, with 
$|E(H)| = \widetilde O\!\left(n^{4/3}\sigma^{2/3}\right)$ edges in running time $\Oish(\sigma^{1/3} n^{2/3}m + \sigma m + \sigma^2 n)$.
\end{corollary}

\Cref{thm:fast-sourcewise-preservers} and \Cref{corr:preserver} recover the state-of-the-art combinatorial bounds for $(1+\varepsilon)$-approximate distance preservers in directed graphs for $p \ge n^{3/4}$ (due to \cite{KP22} and \cite{bernstein_closing_2023}) but with a faster construction time. Indeed, directly implementing the hopset-to-preserver reduction of \cite{KP22} incurs an $\Omega(W n^{\omega+1})$ running time, where $W$ is the graph's aspect ratio, stemming from its reliance on a hop-bounded APSP subroutine \cite{chechik2020distance,zwick2002all}. In contrast, our preserver construction achieves time complexity $O(mn^{\alpha})$, where $\alpha < 1$, for suitably small choices of $\sigma$ and $p$.

\subsection{Summary of Technical Contributions}

\subparagraph{Faster Hopsets.}
To exploit the improved DAG projections to find a faster algorithm for hopsets, we must find a faster algorithm to find the distance analogue of the  chain covers used in shortcut sets.\footnote{See \Cref{def:chain-cover} for a precise definition of these covers.} 
In particular, we need to find chains that are shortest chains (meaning their total weight is equal to the distance between their endpoints).
This makes the min-cost max-flow reduction used in the shortcut set setting inapplicable.
A greedy chain cover algorithm running in $\cO(mn)$ time for the case where the chains do not need to be shortest chains is folklore (see for example Section 3 of \cite{grandoni_all-pairs_2021}).
We extend it to the case with shortest chains in the same running time.
We then speed this up by adding a simple sampling argument allowing us to compute the necessary cover in $\Oish(n/\beta \cdot m)$ time (compared to the $\Oish(mn)$ algorithm relying on APSP).

\subparagraph{Deterministic Shortcut Sets.}
The Kogan-Parter~\cite{KP22a} shortcut set algorithm makes random choices in two of its steps: it samples a set of vertices and a set of chains from a chain cover.
We can avoid the vertex sampling and only keep the chains: one can use the union of the vertices in the chains of the subsampled chains in place of the sampled vertices and the analysis goes through with minimal adjustment.

Picking the subset of chains deterministically, on the other hand, poses significant technical challenges. We want to choose a set of chains that hit the length-$\beta$ prefix and suffix of any path in our graph $G$. This can be interpreted as a set cover instance.
The known random solution can be turned into a deterministic upper bound on the size of this cover using a probabilistic argument.
Of course, while this set cover directly yields a deterministic solution, it is not efficient (not even poly-time).
Observe that we can tolerate a $\cO(\log n)$-approximation of this set cover instance and thus use the greedy set cover algorithm, but even then the instance is too large to be efficiently solved.

Instead we find this solution implicitly using the specific structure of the shortcut set construction.
In particular, we do not need to cover the prefix and suffix of all paths, but can fix a canonical path between any pair of reachable nodes.
We choose these paths such that they have certain desirable properties, based on consistent DFS trees corresponding to the chain cover (see \cref{lem:path-consistent-dfs} for details).

As a final step, we efficiently compute how much each new candidate chain contributes to our set cover solution.
Crucially, to obtain an efficient algorithm, we develop a fast way to dynamically update this information as we greedily extend our partial solution.
This is only possible by the combinatorial properties of the selected canonical paths mentioned above.

\subparagraph{Faster Distance Preservers.} Hopset constructions are known to imply distance preserver constructions by a blackbox reduction due to \cite{KP22}. Thus, our faster hopset algorithm raises the possibility of obtaining a faster  algorithm for constructing distance preservers by plugging our hopset algorithm into this hopset-to-preserver reduction. Unfortunately, the reduction itself includes computational steps like all-pairs shortest paths that is too slow. To obtain the faster $\Oish(n/\beta \cdot m)$ running time, we ``whitebox'' the reduction and crucially use the assumption that the hopsets are computed using specifically our  faster hopset algorithm. In particular, for two consecutive hopsets in the hierarchy $H_i, H_{i+1}$, we build a data structure that can, for each hopset edge $e \in H_{i+1}$, provide a $(1+\smol)$-approximate $\beta_i$-hop path $\Pi(e)$ in $G \cup H_i$ in time near-linear in the output, that is in $\Oish(|\Pi(e)|)$ time.

Finding these witnessing paths is the core computational bottleneck in the reduction. The speed up relies on finding a small enough set of trees encoding the low-hop paths guaranteed by the Kogan-Parter style construction.
To construct them efficiently, we start with a shortest path-tree for a smaller set of relevant vertices.
For any shortest path $P$, we can obtain an approximate shortest path $P'$ obeying the hopbound by finding the first sampled vertex $u$ and the last sampled chain $Q$ on $P$ and rerouting via the $u$ to $Q$ edge added during the hopset construction.
Interestingly, we can do an analogous transformation on an entire shortest path tree, interpreting this rerouting as rehanging parts of the tree to reduce its height.

\section{Preliminaries}
A \emph{directed acyclic graph (DAG)} is a directed graph that does not contain a directed cycle.

\begin{lemma}[DAG SSSP DP, \cite{cormen2022introduction} Sec. 22.2]
\label{lem:dag-sssp}
    There is a $\cO(n+m)$ time algorithm that given a real-valued DAG $G=(V,E,w)$ and source $s \in V$, finds all out-distances from $s$. It can also report a shortest path tree from $s$ in the same time.
\end{lemma}

\begin{definition}[$\ell$-Chain Cover]
\label{def:chain-cover}
    Let an \emph{$\ell$-chain cover} be a set of at most $\ell$ vertex-disjoint chains $\chains$ in $G$ such that for any path $P \subseteq G$, it holds that $\setsize{V(P) \cap (V\setminus V(\mathcal C))} \le 2 n/ \ell$.
\end{definition}

Note that this formalizes the output of the \texttt{DisjointChains} algorithm of \cite{kogan_beating_2022}.

The following lemma is a slightly tighter version of Lemma~1.10 from \cite{kogan_beating_2022}.
\begin{lemma}[Lemma 1.3 of \cite{Raskhodnikova10}]
    \label{lem:path-shortcut}
    For any path graph $P_n$, there is a shortcut set of size $\BigO(n \log^* n)$ for $\beta=4$, which can be computed in $\BigO(n \log^* n)$ time.
\end{lemma}

The idea of \emph{supershortcuts} provides a way to shortcut a path to constant diameter.
This is a common subroutine for more general shortcutting algorithms.

\begin{corollary}
    \label{lem:supershortcut}
    For any DAG $G=(V,E)$ and any set of vertex-disjoint chains $\chains$ in $G$, there is a set of edges $H$ s.t. $\setsize{H} = \BigO(n \log^* n)$ and for any $C \in \chains$ and $u,v \in C$ with $u$ before $v$, we have $\dist_{G \cup H}(u,v) \le 4$.
    Write $\supershortcut(\chains)$ for such an $H$.
    It can be computed in $\BigO(|H|)$ time.
\end{corollary}
\begin{proof}
    Add the edges between consecutive nodes of every chain, and apply \Cref{lem:path-shortcut} to every chain.
\end{proof}

\section{Greedy Lightest Chain Cover in $\Oish(n/\beta\cdot m)$}
\label{sec:lightest-chain-cover}

In \Cref{sec:kp-hopset} we present a simple algorithm for constructing $(1+\varepsilon)$-hopsets of DAGs due to Kogan and Parter \cite{KP22a}. This construction has the following main steps, repeated over every distance scale  in the graph: (1) sample a subset $V'$ of vertices, (2) compute a family $\chains$ of well-behaved chains, (3) super-shortcut along each chain in $\chains$, (4) sample a subset $\chains'$ of chains, and (5) connect each vertex in $V'$ to each chain in $\chains'$ when possible. The generic routing and implementation lemmas there, especially \Cref{lem:route-from-sampled-suffix,lem:build-one-scale}, show that once Step~(2) is available, the remainder of the construction goes through exactly as written. Thus, in this section it remains only to compute the chain family described in Step~(2) quickly.

Every step of the construction algorithm in \Cref{sec:kp-hopset} can be easily implemented in  $\Oish(n/\beta \cdot m)$ time, except Step~(2), where we need to take some care. Constructing the greedy chain family will be the main focus of this section.

Take $V'$ as in the construction in \Cref{sec:kp-hopset}, that is a set of vertices of size $n/\beta \cdot C\log n$ taken uniformly at random from $V$. Write $\beta$ for the target hopbound and $L$ for the distance scale we are working at. Then we aim to compute the following object.

\begin{definition}[Lightest Chain Cover for $V'$]
  \label{def:lightest-chain-cover}
  Call $\chains$ a \emph{lightest chain cover} from $V'$ (for $L$, $\beta$, and $\eta$) if
\begin{conditions}
  \item every chain $Q \in \chains$ is a lightest chain (i.e., if $u,v$ are its endpoints, then $w(Q) = d(u,v)$),\label{cond:d0}
  \item for every $Q \in \chains$, $w(Q) \le \eta L$,\label{cond:d1}
  \item $\setsize{\chains} = \cO(n/\beta)$, where the hidden constant may depend on $\eta$, \label{cond:d2}
  \item chains in $\chains$ are pairwise vertex-disjoint, and
  \item there is no lightest path $P$ starting at any $v' \in V'$ with (a) more than $\beta$ vertices not in $\chains$, and (b) $w(P) \le L$.\label{cond:d3}
\end{conditions}
\end{definition}

By \Cref{lem:prefix-sampling}, every long shortest path reaches the sampled set $V'$ within its first $\cO(\beta)$ vertices. Hence, once we control the remaining suffix by a cover as above, \Cref{lem:route-from-sampled-suffix} from \Cref{sec:kp-hopset} gives the desired low-hop route. Let us now see how to compute this chain cover efficiently.

\begin{lemma}
\label{lem:fast_chain_alg}
    Given $G$, $\beta$, $L$, and the sampled set $V'$ above, there is an algorithm computing a set $\chains$ satisfying \Cref{def:lightest-chain-cover} in $\Oish(n/\beta \cdot m)$ time.
\end{lemma}

Note that in \Cref{alg:lightest-chain-cover} we assume positive integer weights. Then the whole-number part of the found distances encodes the distance by weight, and the decimal part encodes the maximum number of uncovered vertices on any shortest path, allowing us to check the if-condition. For arbitrary nonnegative real weights, the DP of \Cref{lem:dag-sssp} instead compares pairs (distance, $-$\#uncovered vertices) lexicographically. This yields the same paths without any rounding, so all chains stay exactly lightest.

\begin{algorithm}[tbhp]
        \KwIn{$G=(V, E, w)$, a directed acyclic graph with positive integer weights, a set $V' \subseteq V$, and parameters $L,\beta$}

         $\chains \gets \{\}$

        Assign weight $-1/(2n)$ to every vertex

        \For{$v' \in V'$}{
            \Repeat{the if-condition fails}{
                Compute an SSSP from $v'$ (\Cref{lem:dag-sssp})

                \If{there is a node $u$ such that its distance from $v'$ is at most $L$ and it uses at least $\beta$ uncovered vertices}{
                    Add all nodes from the lightest $v'$-$u$ path with negative weight as a chain to $\chains$

                    Set the weight of these nodes to $0$
                  }
                }
        }

        Truncate all chains into pieces of length at most $\eta L$
    \caption{Computing a lightest chain cover for $V'$, $L$, and $\beta$. \label{alg:lightest-chain-cover}}
\end{algorithm}

\begin{proof}[Proof of \Cref{lem:fast_chain_alg}]
For correctness, fix one execution of the repeat-loop for some $v' \in V'$. The modified length of a path is its ordinary edge-weight plus the sum of the current vertex weights on that path. Every vertex weight always lies in $\{0,-1/(2n)\}$, so the total perturbation on any path has absolute value less than $1/2$. Since the edge weights are positive integers, every modified shortest path from $v'$ is therefore also a lightest path in the original graph, and among lightest paths it maximizes the number of currently uncovered vertices.

Hence, whenever the if-condition succeeds, the algorithm finds a lightest path $R$ from $v'$ to some node $u$ with $w(R) \le L$ and at least $\beta$ currently uncovered vertices. The negative-weight vertices on $R$, in their order along $R$, form a lightest chain, and this is exactly the chain added to $\chains$. Afterwards those vertices are assigned weight $0$, so no later chain can contain them. Thus the chains are pairwise vertex-disjoint. After the final truncation, every piece is still a lightest chain, every piece has weight at most $\eta L$, and the covered vertex set is unchanged. Since each successful iteration covers at least $\beta$ new vertices, there are at most $n/\beta$ such iterations; the truncation increases the number of chains by only an $O_\eta(1)$ factor. This proves \Cref{cond:d0,cond:d1,cond:d2} and the pairwise-disjointness condition.

For \Cref{cond:d3}, fix $v' \in V'$. When the repeat-loop for $v'$ stops, there is no lightest path starting at $v'$ with weight at most $L$ and at least $\beta$ currently uncovered vertices. Afterwards the algorithm only adds more chains, so the set of uncovered vertices can only shrink. Therefore the same statement still holds for the final family $\chains$.

For the running time, every successful iteration covers at least $\beta$ previously uncovered vertices, so there are at most $n/\beta$ successful SSSP computations. In addition, each $v' \in V'$ contributes at most one final failed SSSP computation. Thus the running time is $\Oish((n/\beta + |V'|)\cdot m)$, which is $\Oish(n/\beta \cdot m)$ for the sampled set $V'$ used here. The final truncation takes only linear time in the total chain size.
\end{proof}

See \Cref{sec:fast-all-params} for how to plug this greedy algorithm into the hopset construction for all parameter regimes.
This yields the following main theorem.

\hsalldiamthm*

\section{Near Linear-Time Deterministic \(\sqrt{n}\)-Shortcut Sets}
\label{sec:linear-time-deterministic}
In \cite{balsbodwin2025greedy}, the authors observed that combining previous results from \cite{kogan_beating_2022} and \cite{caceres_minimum_2023} yields an almost-linear time algorithm to compute a chain cover.

\begin{lemma}[Lemma 6.1 of \cite{balsbodwin2025greedy}]
    \label{lem:fast-ell-cover}
    We can deterministically compute an $\ell$-chain cover $\chains$ in $\BigO(m^{1+o(1)})$ time.
\end{lemma}

This directly gives a deterministic algorithm to compute \(\Oish(n)\)-size shortcut sets for \(D = \BigO(\sqrt{n})\) in the same time.

The fastest known algorithm for $\sqrt n$-shortcut sets is \texttt{FasterShortcutSqrtN} from \cite{kogan_beating_2022}, running in \(\Oish(m + n^{3/2})\) time.
On the flip side, the best known near-linear time algorithm for shortcut sets is by Cao, Fineman, and Russel \cite{CFR20}.
It achieves a diameter of $\BigO(n^{1/2+o(1)})$.

This results ties the quality of folklore shortcut sets with a near-linear time deterministic algorithm.

\begin{corollary}
    There is a deterministic $\BigO(m^{1+o(1)})$ time algorithm to compute a $\BigO(n \log^* n)$-size $\BigO(\sqrt{n})$-shortcut set for any DAG $G=(V,E)$.
\end{corollary}

\begin{proof}
    The algorithm uses \cref{lem:fast-ell-cover} to compute a $\sqrt n$-chain cover $\chains$ for $G$.
    It then uses \cref{lem:supershortcut} to super-shortcut all chains in $\chains$.
    It directly follows that the algorithm has the desired running time and produces a set of edges of the desired size.

    To see that it is indeed yields a $\BigO(\sqrt n)$-shortcut set, take any pair of reachable nodes $(u,v) \in \tc(G)$.
    Let $\pi$ be a shortest path from $u$ to $v$ in the shortcutted graph.
    By the property of the $\sqrt{n}$-chain cover, there are at most $2\sqrt n$ nodes not on any chain (i.e., from $V \setminus \cup \chains$) on $\pi$.
    Additionally, $\pi$ contains nodes from at most $\sqrt{n}$ chains (since $\setsize{\chains}=\sqrt{n}$) and at most five nodes from each chain.
    If there were more than five nodes from one chain, take the first of these last nodes and substitute their path of length four (which exists by the definition of the super-shortcuts) for the relevant subpath in $\pi$.
    This is a shorter path, showing that the assumption that there are at most five nodes from any chain holds.
    Combining these facts, we see that $\pi$ contains at most $7\sqrt{n}$ nodes.
\end{proof}

\section{Deterministic \(n^{1/3}\)-Shortcut Sets using Greedy Set Cover}
\label{sec:setcover-greedy}
In this section we present a deterministic version of the shortcut set algorithm of~\cite{KP22a}. 
In the process we obtain a number of results we feel are insightful in their own right, revealing more about the connection between $\ell$-chain covers and shortcut sets.

Concretely, to achieve determinism we need to avoid both the sampling of the chains and of the nodes in Kogan and Parter's algorithm in~\cite{KP22a}.
Our solution to this is based on these two observations:
\begin{enumerate}
    \item Instead of sampling from $\chains$ to pick $\chains'$, we can use an implicit set cover instance to efficiently pick a good enough set of chains.
    \item We do not need to sample vertices at all. It is enough to pick the vertices on the set of chains, that is $V' \coloneqq \bigcup \chains'$. 
\end{enumerate}

This yields the main theorem of this section.

\thmSetcover*

In particular, for the state-of-the-art setting $\beta=n^{1/3}$, this yields the following result.
\begin{corollary}
    There is a deterministic
    $\Oish(mn^{2/3})$-time
    algorithm computing $\Oish(n)$-size $\BigO(n^{1/3})$-shortcut sets. 
\end{corollary}

Before we give our concrete construction, we observe the following general structure helps us gain more intuition about why the $\ell$-chain cover is useful for calculating shortcut sets.

\begin{definition}[Path Complex]
\label{def:path-complex}
    Given a DAG $G=(V,E)$, define a \emph{path complex} as a set of paths $\mathcal P$ in $G$ s.t. for any path $\pi \in \mathcal P$, any subpath $\pi'$ of $\pi$ is also in $\mathcal P$.
\end{definition}

Essentially, this allows us to focus on specific paths in the graph, while still being able to obtain shortcut sets guarantees.
This is important to make our set cover efficient as considering all paths or an arbitrarily fixed shortest path between any two nodes would not feasibly lead to a fast algorithm.
The following notion allows us to capture which connections a path complex is supposed to represent.

\begin{definition}
\label{def:path-complex-for-s}
    Given a DAG $G=(V,E)$, define a \emph{path complex for $S \subseteq V \times V$} as a path complex $\mathcal P$ in $G$ s.t. for any pair $(u,v) \in S \cap \tc(G)$, there is a unique path from $u$ to $v$ in $\mathcal P$.
\end{definition}

Intuitively, $S$ encodes the part of the reachability structure of $G$ that should be captured by $\mathcal P$.
This concept is useful because we will effectively show that shortcutting all the paths in an appropriately chosen path complex is enough to shortcut the entire graph.
In particular, we will derive this path complex from an $\ell$-chain cover, and we pick $S$ to be the out-reachability of the chains.
Intuitively, this shows us that the outgoing connectivity from this $\ell$-chain cover is what we need to consider for finding good shortcut sets.
Or in other words: $\ell$-chain covers help computing shortcut sets with smaller diameter because they allow us the focus on a smaller set of connections to shortcut.
Additionally, these connections exhibit strong structure.

For the technical part of this section, fix an $n$-vertex DAG $G=(V,E)$, let $\beta = \BigO(n^{1/3})$ be the target diameter, and fix an $\ell$-chain cover $\chains$ for $G$ with $\ell \coloneqq n/\beta$.
Assuming a DAG is w.l.o.g. by the following folklore observation: in a general digraph, add a bidirected star on every strongly connected component and contract the components, which costs $\BigO(n)$ edges and $\BigO(m)$ time.
Chains are paths in $\tc(G)$, but $\supershortcut(\chains)$ contains the edges between consecutive chain nodes anyway. So we add these at most $n$ edges to $G$; this preserves $\tc(G)$ and the covering property, and makes every chain a path in $G$.
We do not repeat these in the statements of definitions, lemmas, or theorems to improve readability.

Our efficient construction of shortcut sets relies on the fact that the structure outlined in the following lemma compactly represents a path complex for the reachability structure outgoing from $\chains$ (i.e., the union of all paths in the trees form a path complex for $S \coloneqq \cup \chains \times V$).

Since we later prove that shortcutting all paths in this path complex is good enough to obtain a shortcut set in the original graph, this specific path complex from the next lemma is useful as it gives us computationally quick access to the encoded paths and their reachability information. 

\begin{lemma}
    \label{lem:path-consistent-dfs}
    In $\BigO(\ell m)$ time,
    we can compute a set of subtrees $\{T^C\}_{C \in \chains}$ of $G$, each rooted at the first node of the respective chain such that for all $C, C' \in \chains$, $u \in C$, $v \in V$ with $(u,v) \in \tc(G)$, we have
    \begin{enumerate}
        \item there is a path from $u$ to $v$ in $T^C$,
        \item if the path $\pi$ from $u$ to $v$ in $T^C$ uses a node $w$ from $C'$, then $\pi[w,v]$ is a path in $T^{C'}$, and
        \item on the path $u$ to $v$ in $T^C$, any nodes from $C'$ appear consecutively.
    \end{enumerate}
\end{lemma}

This structure stores a path complex (which could contain $\BigO(n^2)$ paths of $\BigO(n)$ length) as a set of $\ell$ subtrees of $G$, that is within $\BigO(\ell n)$ space.
Additionally, performing searches on these trees allows us to answer queries about the outgoing reachability of the chains efficiently.
Note that property (3) also holds for $C=C'$, that is $C$ is always a subpath of $T^C$.

\begin{proof}
    We compute a set of trees with the claimed properties by running a modified DFS from the first node in any chain.
    For this modified DFS, we make two adjustments to the order in which nodes are explored: (A) if there are multiple possible nodes to explore and at least one of them is on the same chain as the current node, explore the one of these nodes that is first in this chain, and (B) otherwise, break ties between possible next choices consistently (e.g., according to a fixed numbering of all nodes).

    Rule (A) ensures that properties (1) and (3) hold and both rules together imply property (2).
    For (1) observe that (A) ensures that during the computation of $C$ any node $u' \in C$ is processed before any node that is reachable from $u'$ in $G$.
    By the definition of a DFS this means that all these nodes will be in the subtree of $T^C$ rooted at $u'$.

    For (3), assume that there are $u \in C$ and $v \in V$ such that the path $\pi$ from $u$ to $v$ in $T^C$ uses nodes from some chain $C'$ non-consecutively. Let $v'_i \in C'$ be the first node from $C'$ that appears on this path and let $v'_j$ be the last.
    When $v'_i$ was processed by the DFS, this has made the node $v'_{i+1}$ available, which by rule (A) must have been processed next.
    Continue this argument inductively to see that $v_i' \rightsquigarrow v_j'$ must be a path in $T^C$.
    Since paths in trees are unique, this means that $\pi$ must contain $v_i' \rightarrow v'_j$ as a subpath, contradicting our assumption that there are non-consecutive nodes from $C'$ in $T^C$.

    For (2), note that rules (A) and (B) fix, for every node, a priority order on its out-neighbors that does not depend on the root of the DFS.
    On a DAG, such a DFS reaches every node first along the lexicographically smallest path w.r.t.\ these priorities (a path along which it is reached earlier would be lexicographically smaller).
    Moreover, prefixes and suffixes of lexicographically smallest paths are again lexicographically smallest (otherwise, exchanging them would yield a smaller path, as $G$ is acyclic).
    Hence, $\pi[w,v]$ and the $w$-$v$ path in $T^{C'}$, which exists by (1), are both the lexicographically smallest $w$-$v$ path, so they coincide. Intuitively, all trees use the same canonical path between any two nodes.
\end{proof}

For the rest of this section, fix the set $\{T^C\}_{C \in \chains}$ described in this lemma.
For $v' \in C$, let $T^C_{v'}$ be the subtree of $T^C$ rooted at $v'$ without the subtree of the successor of $v'$ on $C$. By (1), these are the nodes for which $v'$ is the last node of $C$ that can reach them. In particular, $v'$ is the only node of $C$ in $T^C_{v'}$. If $v'$ is the $i$-th node of $C$, we also write $T^C_i$.

\begin{corollary}
    The union of all paths in the trees $\{T^C\}_{C \in \chains}$ form a path complex for $S \coloneqq \cup \chains \times V$.
\end{corollary}
Path complexes only serve as intuition; the proofs below work directly with the trees $\{T^C\}_{C\in\chains}$.

To achieve the first goal outlined at the beginning of this section (say, for $\beta = n^{1/3}$), we need to pick a subset of the chains \(\chains' \subseteq \chains\) such that (1) \(\setsize{\chains'} = \Oish(n^{1/3})\) and (2) for any two reachable nodes which are far enough apart, there is a path where one of the sampled chains appears in the first \(\BigO(n^{1/3})\) nodes and one of the sampled chains appears in the last \(\BigO(n^{1/3})\) nodes.
Regarding (1), note that due to the probability in the random sampling, in the Kogan and Parter algorithm \(\chains'\) has size \(n^{1/3} \log n\).
Regarding (2), note that the above holds with high probability when each chain is sampled with probability \(\log n / n^{1/3}\) by applying a simple Chernoff bound.
We will instead construe the problem as a set cover instance as follows.
We fix a path complex that encodes all out-connections from $\chains$ (using $\{T^C\}_{C \in \chains}$).
The elements to be covered are the prefixes and suffixes of (normalized, see below) length $\beta$ of these paths, and each chain covers all of them that it intersects.
It is easy to see that this instance is potentially too large to be constructed explicitly.
Luckily, the greedy algorithm for set cover can be executed iteratively and implicitly.

We first prove the following lemma which allows us to ensure that no chain is too long.
This will be important in running time calculations later.

\begin{lemma}
    \label{lem:split-up-chains}
    Given an $\ell$-chain cover $\chains$, there is a set of at most $2\ell$ vertex-disjoint chains $\chains^*$ with $\cup\chains^* = \cup\chains$ s.t. every chain in $\chains^*$ has length at most $n/\ell$.
    Also, $\chains^*$ can be computed from $\chains$ in linear time.
\end{lemma}
\begin{proof}
    Simply take $\chains$ and split every chain with length more than $n/\ell$ into segments (along the topological order of the chain) of length at most $n/\ell$ to obtain $\chains^*$.
    Since the chains are vertex-disjoint, there can be at most $\ell$ chains of length $n/\ell$ in $\chains^*$.
    Splitting up the chains can result in at most one chain in $\chains^*$ with length smaller than $n/\ell$ for every chain in $\chains$. 
    Thus, there are at most $\ell$ of those.
    It is easy to verify that \(\chains^*\) retains the covering property as $\cup \chains = \cup \chains^*$.
\end{proof}

To efficiently compute the set cover instance, we need to make few technical assumptions on which kind of long paths we wish to cover and how we measure their length. 

\begin{definition}
   For any pair $(v,w)$ with $v \in \cup \chains$, define their \emph{normalized distance $d'(v,w)$} as the number of different chains on the path from $v$ to $w$ in $T^C$, where $C$ is the chain of $v$.
   Define the \emph{normalized length} of a path analogously.
\end{definition}
Intuitively, after super-shortcutting, the normalized length of a path determines its hop length up to a constant factor and the at most $2\beta$ nodes outside $\cup\chains$.

From this we can deduce the definition for the set of paths we want to hit with our set cover.

\begin{definition}
    Let $C \in \chains$. A path $\pi$ in $T^C$ is a relevant subpath if $\pi$ has normalized length $\beta$, there is exactly one node from $C$ on $\pi$, namely its first node, the last node from $\pi$ is on a chain from $\chains$ and is the only node from that chain on $\pi$.
\end{definition}

Before we show how to efficiently compute a small set of chains covering all relevant subpaths, let us prove that it in fact yields a shortcut set.
The core of this argument is similar to the analysis of the original Kogan and Parter analysis, but requires some care to respect the path complex given by $\{T^C\}_{C \in \chains}$.
See \Cref{fig:setcover-gives-shortcuts} for an illustration of the resulting short path.

\begin{lemma}
    \label{lem:setcover-gives-shortcuts}
    If $\chains' \subseteq \chains$ is a set s.t. every relevant subpath contains at least one node from a chain in $\chains'$, then $\supershortcut(\chains) \cup \hat H$ for $\hat H \coloneqq \{e(v,C) \mid v \in \cup \chains', C \in \chains'\}$ is a $\BigO(\beta)$-shortcut set of size $\BigO(n \log^* n + \setsize{\chains'}^2 \cdot \beta)$,
    where $e(v,C)$ is an edge from $v$ to the first node reachable from $v$ in $C$.
\end{lemma}

\begin{figure}
    \centering
    \includegraphics[page=3]{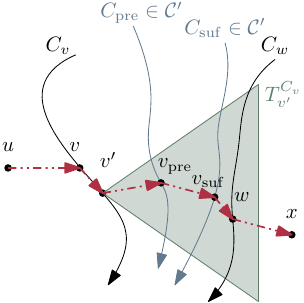}
    \caption{The short path constructed in the proof of \Cref{lem:setcover-gives-shortcuts}. Every red edge represents a segment of the path that has length at most $\beta$ (or even constant for the parts that follow a chain). $C_\pre$ and $C_\suf$ are included in the set cover $\chains'$. The segments between $v'$ and $v_\pre$ as well as between $v_\suf$ and $w$ are the prefix and suffix of $\pi'$ covered by $\chains'$.}
    \label{fig:setcover-gives-shortcuts}
\end{figure}

\begin{proof}
    The size of the shortcut set is given by the size of the supershortcuts (in total $\BigO(n \log^* n)$) and the edges in $\hat H$.
    Note that by \Cref{lem:split-up-chains} we can assume that each chain has at most $\BigO(n/\ell) = \BigO(\beta)$ vertices and thus there are $\BigO(\setsize{\chains'}^2 \cdot n / \ell) = \BigO(\setsize{\chains'}^2 \cdot \beta)$ edges in $\hat H$.

    Let $(u,x) \in \tc(G)$.
    Note that we do not require $u \in \cup\chains$.
    We show that there is a path from $u$ to $x$ of length $\BigO(\beta)$ using the claimed shortcut set.

    Let $\pi$ be any path from $u$ to $x$.
    By definition of the $\ell$-chain cover $\chains$, there are at most $2\beta$ nodes on $\pi$ that are not in $\cup \chains$.
    Let $v$ and $w$ be the first and last nodes from $\cup \chains$ on $\pi$ respectively (if there are none, $\pi$ has at most $2\beta$ nodes).
    Note that we have $d_\pi(u,v) \le 2\beta$ and $d_\pi(w,x) \le 2\beta$ and $(v,w) \in \tc(G)$.
    Let $C_v$ be the chain of $v$ and $C_w$ be the chain of $w$.
    Let $v'$ be the last node on $C_v$ that can reach $w$. 
    Then there is a path $\pi' \colon v' \rightsquigarrow w$ in $T^{C_v}_{v'}$.
    If the $\pi'$ has normalized length less than $2\beta$, we are done.
    Observe that $\pi'$ can contain at most $2\beta$ nodes not from $\cup\chains$ and contains less than $2\beta$ chains.
    Using the edges from $\supershortcut(\chains)$, we can replace any number of consecutive nodes from the same chain with a path of length $4$.
    Thus, there is a path from $v'$ to $w$ with length $\BigO(\beta)$ in the shortcutted graph and thus from $u$ to $x$, there is a path with length at most $\BigO(\beta)$.

    Now consider the case where $\pi'$ has normalized length at least $2\beta$, so its prefix and suffix of normalized length $\beta$ are disjoint.
    Since $\chains'$ hits all relevant subpaths, it hits the prefix and suffix of $\pi'$ with normalized lengths $\beta$ (both are relevant subpaths: the suffix is a path in the tree of its first chain by \Cref{lem:path-consistent-dfs}~(2), and we may assume that $w$ is the only node of $C_w$ on $\pi'$, as by (3) we can move $w$ to the first such node and then reach the original $w$ within $4$ hops along $C_w$), assume it does so by chains $C_\pre, C_\suf \in \chains'$. 
    Then there are $v_\pre \in C_\pre$ and $w_\suf \in C_\suf$ s.t.
    \begin{enumerate}
        \item $(v_\pre, w_\suf)$ is an edge in $\hat H$,
        \item $v_\pre$ is on $\pi'$,
        \item $w_\suf$ is at most as late on $C_\suf$ as some node $y \in C_\suf$ on the suffix of $\pi'$.
    \end{enumerate} 

    From this, we can conclude that 
    \begin{align*}
    \pi'' \coloneqq
        u \overset{V \setminus \chains}{\rightsquigarrow}
        v \overset{C_v}{\rightsquigarrow}
        v' \overset{\beta\text{-prefix in } T^{C_v}_{v'}}{\longleadsto{1.9}}
        v_\pre \overset{\hat H}{\to}
        w_\suf \overset{C_\suf}{\rightsquigarrow} y \overset{\beta\text{-suffix in } T^{C_v}_{v'}}{\longleadsto{1.9}}
        w \overset{V \setminus \chains}{\rightsquigarrow}
        x
    \end{align*}
    is a path in $G \cup \supershortcut(\chains) \cup \hat H$ where every segment has length $\BigO(\beta)$ and so the total length is $\BigO(\beta)$.
    See \Cref{fig:setcover-gives-shortcuts} for an illustration of this path.
\end{proof}

The following lemma is the key to efficiently computing how many prefixes each chain is able to cover.

\begin{lemma}
    \label{lem:leafs-to-prefixes}
    There is a one-to-one correspondence between relevant subpaths starting on a chain $C \in \chains$ and leaves of normalized depth $\beta$ in the trees $\{T^C_i\}_{i \in \setsize{C}}$ if they are truncated to normalized depth $\beta$ (i.e., every root path is cut directly after its first node on the $\beta$-th chain).
\end{lemma}

\begin{proof}
    This follows directly from the definition.
    Let $\pi$ be a relevant subpath and $C_1, \dots, C_\beta$ the chains on it.
    Write $v_1$ for the only node from $C_1$ on $\pi$ and $v_\beta$ for the only node from $C_\beta$ on $\pi$.
    Then it is easy to see (1) $v_\beta$ is a leaf in $T^C_{v_1}$ and (2) that it is the only relevant subpath containing this leaf, yielding the bijection.
\end{proof}

Before we proceed, we give a formal definition of the potential underlying the set cover instance.

\begin{definition}
    Given a subset of the chains $\chains' \subseteq \chains$, a chain $C \in \chains$ and a node $v \in \cup \chains$ from a chain $C_v$, write $\coverCount^C(v)$ for the number of relevant subpaths that are (1) not covered by $\chains'$, (2) originate in $C$, and (3) on which $v$ is the first node from $C_v$ that appears.
    Similarly, define $\coverCount(C)$ as the number of such uncovered relevant subpaths on which a node from $C$ appears.
    Call this the \emph{covering potential of nodes and chains}, respectively.
\end{definition}
Observe that in this counting all the prefixes originating on a specific chain $C$ are counted within $\coverCount^C(v_1)$ where $v_1$ is their first node.

While \Cref{lem:leafs-to-prefixes} will help us count how many prefixes can be covered using a specific node, the following lemma allows to deduce how useful an entire chain is from that information.

\begin{lemma}
    \label{lem:cover-count:nodes-to-chains}
    For any chain $C \in \chains$, we have that $\coverCount(C) = \sum_{v \in C} \sum_{C' \in \chains} \coverCount^{C'}(v)$.
\end{lemma}

\begin{proof}
    Take an arbitrary uncovered relevant subpath hit by $C$.
    Assume it originates on $C'$ and $v$ is the first node on $C$ that appears in the relevant subpath.
    Then this relevant subpath is counted precisely in $\coverCount^{C'}(v)$ and in no other element of the right-hand side of the equation. 
\end{proof}

To execute the greedy algorithm for set cover efficiently, we need two important subroutines: (1) compute $\varphi^C(v)$ for all $C \in \chains$ and $v \in \cup \chains$, and (2) given a chain $C$ to include into the partial set cover $\chains'$, update these values.

The following lemma shows that the first of these tasks can be achieved in $\BigO(\ell m)$ time.
See \Cref{fig:set-cover-counting} for an illustration of the resulting data structure.

\begin{lemma}
    \label{lem:initial-potential}
    For the empty set of picked chains $\chains' = \emptyset$, we can compute all $\varphi^C(v)$ and $\varphi(C)$ values
    in $\BigO(\ell m)$ time.
\end{lemma}

\begin{figure}
    \centering
    \includegraphics[page=2]{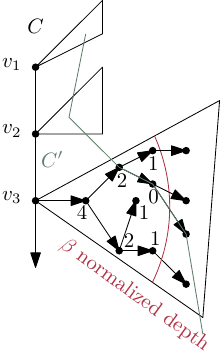}
    \caption{Result of computing $\coverCount^C(v)$ for the nodes in $T_3^C$. Observe that multiple nodes from the chain $C'$ appear in $T_3^C$ and only those occurrences that are not a descendant of another node from $C'$ in $T^C_3$ are assigned non-zero values.}
    \label{fig:set-cover-counting}
\end{figure}

\begin{proof}
    Using \Cref{lem:path-consistent-dfs}, we can compute the trees $\{T_i^C\}_{C \in \chains, i \in \setsize{c}}$
    in $\BigO(\ell m)$ time. 
    These we can truncate to normalized depth $\beta$ in time linear in their size, so $\BigO(\ell n)$ time.

    Additionally, for each such tree, say $T^C_i$ with root node $v_i$, we compute for each node $v \in T^C_i$ whether there is a node from the chain of $v$ on its paths from the root $v_i \to v$.
    Since all nodes from a chain must appear consecutively in the tree, this can be simply done by checking if $v$'s predecessor is on the same chain.
    In total, this again requires $\BigO(\ell n)$ time.

    Now, given a tree $T^C_i$, for each node $v$ that does not have a node from its chain as a predecessor in $T^C_i$, we can calculate $\coverCount^C(v)$ as the number of leafs of normalized depth $\beta$ in the subtree rooted at $v$.
    Using the standard dynamic program on trees, this can be done in $O(\setsize{T^C_i})$ time per tree, so again $\BigO(\ell n)$ time in total.

    This calculation is correct by \Cref{lem:leafs-to-prefixes}.
    See \Cref{fig:set-cover-counting} for an illustration of the result.
    
    As by \Cref{lem:split-up-chains}, we may assume that every chain has at most $\BigO(n/\ell)$ nodes and thus using \Cref{lem:cover-count:nodes-to-chains} we can compute the values for $\coverCount(C)$, that is the covering potential of the chains in $\BigO(\ell \cdot n)$ time (each node is on exactly one chain and we have to sum up its covering potential wrt. $\ell$ chains).
\end{proof}

\begin{lemma}
    \label{lem:update-potential}
    We can process a series of $\opCount$ operations of the following type in $\BigO(\ell n + \opCount \beta n)$ time:
    Given a set of picked chains $\chains'$, and corresponding values $\{\coverCount^C(v)\}_{C \in \chains, v \in \cup \chains}$ and $\{\coverCount(C)\}_{C \in \chains}$, and a chain $C \in \chains \setminus \chains'$, update the $\coverCount^C(v)$ and $\coverCount(C)$ values wrt. $\chains' \cup \{C\}$.
\end{lemma}

\begin{proof}
    \emph{Preprocessing.}
    By a linear scan through all trees $\{T^C\}_{C \in \chains}$, precompute for each node $v \in V$ and each chain $C \in \chains$ the node $v' \in C$ s.t. $v \in T^C_{v'}$ (i.e., the last node on $C$ that can reach $v$).
    This can be stored in a simple lookup table $\treefinder$ of size $\BigO(\ell n)$.
    Additionally, for any $C$ and $C' \ne C$, if there are nodes $u,v \in C'$ s.t. $u$ is the direct parent of $v$ in $T^C$, then contract them into a single node.
    Note that in this case we had $\coverCount^C(v) = 0$ by definition.
    This changes nothing about the connectivity structure of $T^C$ and only removes some redundant nodes.
    This compaction can be implemented via a DFS through the trees and thus in total time $\BigO(\ell n)$.

    \emph{Update operation.}
    Relying on \Cref{lem:leafs-to-prefixes}, we will find all nodes $v$ whose covering potential $\varphi^{C'}(v)$ for some $C'$ needs to be updated by finding those whose covering potential depends on leaves in some $T^{C'}$ that are newly covered by $C$.
    We assume the trees are truncated to normalized depth $\beta$.
    \begin{algorithm}[tbhp]
        \DontPrintSemicolon
        \SetKwFunction{FMain}{UpdateCoveringPotential}
        \SetKwProg{Fn}{fun}{:}{}
    For every $u \in C$, $C' \in \chains$:

    1. Find the unique $v$ s.t. $u \in T^{C'}_v$ using $\treefinder$ (skip if there is none).

    2. For all ancestors $v'$ of $u$, we  reduce their value by $\coverCount^{C'}(u)$.

    3. For all descendants $v'$ of $u$ in $T^{C'}_v$ (including $u$), set their $\coverCount^{C'}(v')$ value to 0.

    4. Delete $u$ and its subtree from $T^{C'}$, updating $\treefinder$ in the process.

    \caption{Updating the covering potentials $\coverCount^{C'}(v)$ via the truncated trees $T^{C'}$.
        \label{alg:update-covering-potential}}
    \end{algorithm}

    To perform the updates, use \Cref{alg:update-covering-potential}.
    Note that through $u$ the chain $C$ covers all relevant subpaths corresponding to the leafs in the truncated subtree of $u$ in $T^{C'}_v$ (for $C'=C$, these are all relevant subpaths starting at $u$).
    These are counted precisely for the non-zero values $\coverCount^{C'}(v')$ of $v'$ in $T^{C'}_v$ such that $v'$ is an ancestor or a descendant of $u$ in $T^{C'}_v$.
    Also note that every time we update a covering potential for a node, we can easily update the covering potential for its chain in the process at no asymptotic increase to the running time.

    Let us now count the amortized cost of $\opCount$ of these operations.
    The loop at the beginning of the algorithm runs $\BigO(\beta \cdot \ell) = \BigO(n)$ times.
    Step 1 can be implemented in constant time since $\treefinder$ is precomputed.
    For Step 2, observe that since the trees are truncated to have normalized depth at most $\beta$, the path from $u$ to the root contains nodes from at most $\beta$ different chains and by the compaction above, only one node per chain.
    Additionally, it contains at most $\BigO(\beta)$ nodes in $V \setminus \cup \chains$.
    So $u$ has at most $\BigO(\beta)$ ancestors for which need to reduce their covering potential.
    Therefore, in total the Steps 1 to 2 take at most $\BigO(\beta)$ time per iteration of the loop, i.e., $\BigO(\beta n)$ per operation.

    For Step 3 and 4, we amortize over all operations.
    Since any node in any of the trees $T^C$ is deleted at most once, we incur total time at most $\BigO(\ell n)$.
\end{proof}

Now we have all pieces in hand to prove our main theorem.

\thmSetcover*

\begin{proof}
    Compute an $\ell$-chain cover by \Cref{lem:fast-ell-cover} and split it by \Cref{lem:split-up-chains}. Use \Cref{lem:path-consistent-dfs} to compute the trees $\{T^C\}_{C \in \chains}$, then use \Cref{lem:initial-potential} to compute the initial covering potential of all nodes and chains.
    Adding the running times from these lemmas shows that this setup time is within the claimed bounds. 

    Now we are ready to solve our set cover instance: Iterate over all chains to find the one with the maximal potential.
    This takes $\BigO(\ell)$ time.
    Then use the algorithm from \Cref{lem:update-potential} to update the covering potential of all nodes and chains.
    Repeat this process until no chain has positive covering potential anymore.
    Let the resulting set of chains be $\chains'$.

    We show that $\setsize{\chains'} = \BigO(n/\beta^2 \log n)$.
    Consider the following fractional solution to the set cover instance for covering all relevant subpaths with chains from $\chains$.
    Simply assign each chain the fractional value $\beta^{-1}$.
    Since each relevant subpath contains exactly $\beta$ different chains, the sum of the chains covering it, is $1$.
    Thus, this fractional assignment solves the set cover instance with total weight $\BigO(\ell \cdot \beta^{-1}) = \BigO(n/\beta^2)$.
    It is well known that the greedy solution to a set cover instance is within a {$\log$}-factor from the fractional optimum. 
    Since there are at most $n^2$ relevant subpaths, this $\log$-factor is at most $2\log n$.
    Thus, we have $\setsize{\chains'} = \BigO(n/\beta^2 \log n)$. 

    Plugging this into \Cref{lem:update-potential}, shows that we can process all updates in $\BigO(\ell m + n/\beta^2 \log n \cdot \beta \cdot n) = \BigO(mn/\beta + n^2 \log n / \beta)$ time.

    Plugging this into \Cref{lem:setcover-gives-shortcuts}, give size and quality of the shortcut set.
    In particular, the size is $\BigO(n \log^* n + \setsize{\chains'}^2 \cdot \beta) = \BigO(n \log^* n + (n/\beta^2 \log n)^2 \cdot \beta) = \BigO(n \log^* n + n^2/\beta^3 \cdot \log^2 n) \subseteq \Oish(n^2 / \beta^3)$, giving the desired bound.
\end{proof}

Note that the running time is in $\BigO(nm/\beta)$ instead of $\Oish(nm/\beta)$ if $\beta \in \Omega(\log n)$ and $m \in \Omega(n \log n)$, which is the interesting setting.
This is the same asymptotic bound (including the same $\log$ factors) as in \cite{KP22a}.
We believe that the techniques in this section can also be used to derandomize the algorithm from \cite{KP22a} for the regime where $\beta = \Omega(n^{1/3})$, matching the known bounds again, but we do not prove this here.

\section{Fast Sourcewise Approximate Distance Preservers}

\label{sec:fast-sourcewise-preservers}

The goal of this section is to present our fast algorithms for sourcewise approximate distance preservers.

\FastSourcewisePreservers*

We will prove this theorem by combining the fast hopset constructions from
\Cref{thm:hs-small-diam,thm:hs-large-diam} with the hopset-to-preserver
reduction of Kogan and Parter~\cite{KP22}. The only point in that reduction that is too
slow for our setting is the call to hop-bounded all-pairs shortest paths. We replace that APSP call in two different ways.

First, for internal hopset edges, we design a data structure that, when queried with a hopset edge $(u,v)$, outputs a low-hop $(1+\varepsilon)$-approximate shortest $(u,v)$-path in the underlying graph, which we call a \emph{witness path}. This data structure is summarized in \Cref{lem:witness-paths-for-hopset-edges}.

\begin{restatable}[Fast witness paths for hopset edges]{lemma}{WitnessPathsForHopsetEdges}
\label{lem:witness-paths-for-hopset-edges}
Let $G = (V, E, w)$ be an $n$-vertex DAG with $m$ edges, nonnegative edge weights, and polynomial aspect ratio. Let $H_i, H_{i-1} \subseteq \tc(G)$  be $(1+\delta)$-approximate hopsets for $G$ produced by
\Cref{thm:hs-small-diam} or \Cref{thm:hs-large-diam}, with hopbounds at most $\beta_i$ and $\beta_{i-1}$, respectively.
Then in $\widetilde O(n/\beta_i \cdot m)$ preprocessing time, we can equip hopset $H_i$ with a data structure guaranteeing that every  edge $e=(u,v)\in H_i$ can be expanded in $\widetilde O(|W(e)|)$ time into a path 
$W(e) \subseteq G \cup H_{i-1}$
of at most $\beta_{i-1}$ edges and weighted length at most $(1+\delta)d_G(u,v)$. 
Moreover, for any queried subset $F \subseteq H_i$, the total time to report all witness paths
$\{W(e) : e \in F\}$ is
$\widetilde O\!\left(\sum_{e\in F} |W(e)|\right).$
\end{restatable}
\begin{proof}
We defer the proof of this lemma to \Cref{app:hopset_data_structs}.
\end{proof}

Second, for each source $s \in S$ in our source-restricted distance preservers, we need to compute hop-bounded single-source shortest paths rooted at $s$. We prove this can be computed efficiently in \Cref{lem:rooted-path-retrieval}.
\begin{restatable}[Fast rooted path retrieval for hopsets]{lemma}{RootedPathRetrieval}
\label{lem:rooted-path-retrieval}
Let $H \subseteq \tc(G)$ be a $(1+\delta)$-approximate hopset of a graph $G$ produced by \Cref{thm:hs-small-diam} or \Cref{thm:hs-large-diam} with hopbound at most $\beta$. Then for every root $r \in V$, one can in $\widetilde O(m + |H|)$ time compute, for every vertex $v$ reachable from $r$, a path $\Pi_H(r,v) \subseteq G \cup H$ with at most $\beta$ edges and weighted length $w(\Pi_H(r,v)) \le (1+\delta)d_G(r,v)$.
\end{restatable}
\begin{proof}
We defer the proof of this lemma to \Cref{app:hopset_data_structs}.
\end{proof}

We are now ready to prove \Cref{thm:fast-sourcewise-preservers}.

\begin{proof}[Proof of \Cref{thm:fast-sourcewise-preservers}]
If $p=0$, then the empty subgraph is a valid output. Hence assume $p\ge 1$.
For each $s\in S$, let $T_s := \{t\in V : (s,t)\in P\}$, so $\sum_{s\in S}|T_s|=p$.
All hopbound parameters below are rounded up to integers. This changes all bounds only by constant factors. Also, $\beta_i$ denotes the parameter given to \Cref{thm:hs-small-diam,thm:hs-large-diam}, so the actual hopbound of $H_i$ and the lengths of the paths in \Cref{lem:witness-paths-for-hopset-edges,lem:rooted-path-retrieval} are $\cO(\beta_i)$.

\subparagraph{ Sourcewise reduction from a collection of hopsets~\cite{KP22}.}
Let $\eta>0$ be a positive real parameter that we will choose later.
Let $H_1,\ldots,H_L$ be a hopset hierarchy\footnote{A hopset hierarchy is just a collection of hopsets with different hopbounds that we will specify later.} of $(1+\eta)$-approximate hopsets for $G$, where $H_i$ has hopbound at most $\beta_i$ and $\beta_i \ge \beta_{i+1}$ for $i \in [1, L-1]$. Set $H_0:=\emptyset$ and $\beta_0:=n$. Assume that for each $i\in[L]$, the hopset $H_i$ is equipped with the witness-path data structure of \Cref{lem:witness-paths-for-hopset-edges} using $H_{i-1}$ as its auxiliary hopset (for $i=1$, as $H_0=\emptyset$, the witness paths are shortest paths in $G$), and that the last level $H_L$ is also equipped with \Cref{lem:rooted-path-retrieval}.

We build a subgraph $G_P\subseteq G$ as follows.
Intuitively, we route every demand pair through $H_L$ and then repeatedly replace each hopset edge by its witness path one level lower, until only $G$-edges remain.
For every source $s\in S$ with $T_s\neq\emptyset$, invoke \Cref{lem:rooted-path-retrieval} on $H_L$ rooted at $s$. This yields, in $\widetilde O(m+|H_L|)$ time, paths $\Pi_{H_L}(s,v)\subseteq G\cup H_L$ for all vertices $v$ reachable from $s$, with $|\Pi_{H_L}(s,v)|\le \beta_L$ and $w(\Pi_{H_L}(s,v))\le (1+\eta)d_G(s,v)$.
For every demand pair $(s,t)\in P$ with $t$ reachable from $s$, traverse $\Pi_{H_L}(s,t)$, insert every edge of $\Pi_{H_L}(s,t)$ that belongs to $G$ into the output graph $G_P$, and insert every edge of $\Pi_{H_L}(s,t)$ that belongs to $H_L$ into a set $F_L\subseteq H_L$.

Then, for $i=L,L-1,\ldots,1$, process every edge $e\in F_i$ once. Query the witness-path data structure of level $i$ for $e=(u,v)$, obtaining a path $\Pi_i(e)\subseteq G\cup H_{i-1}$ with at most $\beta_{i-1}$ edges and $w(\Pi_i(e))\le (1+\eta)d_G(u,v)$. Insert every $G$-edge\footnote{Throughout, a $G$-edge simply means an edge of the original graph $G$, as opposed to a hopset edge.} of $\Pi_i(e)$ into the output, and insert every $H_{i-1}$-edge of $\Pi_i(e)$ into a set $F_{i-1}\subseteq H_{i-1}$.
Let $G_P$ be the union of all inserted $G$-edges.

We now prove correctness. For every $i\in[L]$ and every $e=(u,v)\in F_i$, define recursively a $u$-$v$ path $\widehat \Pi_i(e)\subseteq G$ by fully expanding every hopset edge in $\Pi_i(e)$ down to level $0$. We claim that
\[
    w(\widehat \Pi_i(e)) \le (1+\eta)^i d_G(u,v)
    \qquad \text{for every } e=(u,v)\in F_i.
\]
The proof is by induction on $i$.

For $i=1$, we have $\Pi_1(e)\subseteq G$ because $H_0=\emptyset$, and therefore
\begin{equation*}w(\widehat \Pi_1(e)) = w(\Pi_1(e)) \le (1+\eta)d_G(u,v).\end{equation*}

Assume the claim holds for level $i-1$, and let $e=(u,v)\in F_i$. Each edge $f$ on $\Pi_i(e)$ is either a $G$-edge or an edge of $H_{i-1}$. If $f\in G$, we keep it as is. If $f=(x,y)\in H_{i-1}$, then $f\in F_{i-1}$ by construction, so we may replace it by $\widehat \Pi_{i-1}(f)$. Since hopset edges are weighted by exact $G$-distances, $w(f)=d_G(x,y)$; hence the induction hypothesis gives $w(\widehat \Pi_{i-1}(f)) \le (1+\eta)^{i-1}w(f)$. Because $(1+\eta)^{i-1}\ge 1$, every edge of $\Pi_i(e)$ contributes at most a factor $(1+\eta)^{i-1}$ to the weight after full expansion. Therefore,
\[
    w(\widehat \Pi_i(e))
    \le (1+\eta)^{i-1} w(\Pi_i(e))
    \le (1+\eta)^i d_G(u,v).
\]

Now fix a pair $(s,t)\in P$ with $t$ reachable from $s$, and let $\widehat\Pi(s,t)\subseteq G$ be the path obtained by fully expanding every $H_L$-edge of $\Pi_{H_L}(s,t)$ using the corresponding paths $\widehat \Pi_L(\cdot)$. By the same argument as above,
\[
    w(\widehat\Pi(s,t))
    \le (1+\eta)^L w(\Pi_{H_L}(s,t))
    \le (1+\eta)^{L+1} d_G(s,t).
\]
All edges of $\widehat\Pi(s,t)$ are inserted into $G_P$ by construction. Thus $G_P$ contains an $s$-$t$ path of this weight. If $t$ is not reachable from $s$ in $G$, then it is also unreachable in the subgraph $G_P\subseteq G$. Consequently, if
\begin{equation*} (1+\eta)^{L+1}\le 1+\varepsilon, \end{equation*}
then $G_P$ is a $(1+\varepsilon)$-approximate distance preserver for $P$.

The work done after these hopsets have been built is
\[
\widetilde O\!\left(
\sum_{s\in S}(m+|H_L|)
   + \sum_{(s,t)\in P} |\Pi_{H_L}(s,t)|
   + \sum_{i=1}^L \sum_{e\in F_i} |\Pi_i(e)|
\right),
\]
where the last term is exactly the reporting bound of \Cref{lem:witness-paths-for-hopset-edges}. Because every top-level path has at most $\beta_L$ edges, and because $|F_i|\le |H_i|$ and every witness path has at most $\beta_{i-1}$ edges, this becomes
\[
\widetilde O\!\left(
\sigma(m+|H_L|) + p\beta_L + \sum_{i=1}^L |H_i|\beta_{i-1}
\right).
\]
The output edge count is bounded by the same expression without the first summand:
\begin{equation*} |E(G_P)| \le p\beta_L + \sum_{i=1}^L |H_i|\beta_{i-1}. \end{equation*}

It remains to choose the hopbounds and corresponding hopsets.

\subparagraph{Case 1: $p\le n$.}
\begin{equation*} \beta := n/p^{2/3}. \end{equation*}
Since $p\le n$, we have $n^{1/3}\le \beta\le n$.
Construct a nonincreasing sequence $\beta_0=n\ge \beta_1 > \cdots > \beta_L = \beta$ as follows: while the current value $\beta_{i-1}$ is larger than $2\beta$, set
\begin{equation*} \beta_i := \beta_{i-1}^{2/3}\beta^{1/3} \end{equation*}
(chosen so that every level contributes the same amount $|H_i|\beta_{i-1}$ below).
When the current value first becomes at most $2\beta$, append one final level equal to $\beta$.
If $\beta_0,\beta_1,\ldots,\beta_{L-1}$ are the values before the final appended level, then $x_i:=\beta_i/\beta$ satisfies $x_i=x_{i-1}^{2/3}$ for every $1\le i\le L-1$, with $x_0=n/\beta=p^{2/3}$. Hence $L=O(\log\log n)$.

Every level satisfies $\beta_i\in[\beta,n]\subseteq[n^{1/3},n]$, so \Cref{thm:hs-large-diam} gives a $(1+\eta)$-hopset $H_i$ with hopbound at most $\beta_i$, size $|H_i| = \widetilde O\!\left((n/\beta_i)^{3/2}\right)$, and construction time $\widetilde O(n/\beta_i\cdot m)$. By \Cref{lem:witness-paths-for-hopset-edges}, augmenting level $i$ with witness-path queries costs another $\widetilde O(n/\beta_i\cdot m)$ time, which is absorbed in the same asymptotic bound.

For every nonfinal level, the recurrence gives $\beta_i^{3/2} = \beta_{i-1}\beta^{1/2}$, and hence
\[
|H_i|\beta_{i-1}
   = \widetilde O\!\left((n/\beta_i)^{3/2}\beta_{i-1}\right)
   = \widetilde O\!\left(n^{3/2}/\beta^{1/2}\right)
   = \widetilde O\!\left(np^{1/3}\right).
\]
The same bound also holds for the final level because $\beta_L=\beta$ and $\beta_{L-1}=O(\beta)$. Therefore,
\begin{equation*} \sum_{i=1}^L |H_i|\beta_{i-1} = \widetilde O\!\left(np^{1/3}\right). \end{equation*}
Moreover,
\begin{equation*} |H_L| = \widetilde O\!\left((n/\beta)^{3/2}\right) = \widetilde O(p), \end{equation*}
so the rooted retrieval stage costs
\begin{equation*} \widetilde O\!\left(\sigma(m+p)\right). \end{equation*}
Also,
\begin{equation*} p\beta_L = \widetilde O(p\beta)=\widetilde O(np^{1/3}), \end{equation*}
and the total time to build the hierarchy is
\[
\widetilde O\!\left(\sum_{i=1}^L n/\beta_i \cdot m\right)
   = \widetilde O\!\left((n/\beta)\cdot m\right)
   = \widetilde O\!\left(p^{2/3}m\right),
\]
because $L=O(\log\log n)$ and every $\beta_i\ge \beta$.

\begin{equation*} \eta := \varepsilon/(2(L+1)). \end{equation*}
Then $(1+\eta)^{L+1}\le 1+\varepsilon$, so the  reduction is correct. We obtain
\[
|E(G_P)| = \widetilde O\!\left(np^{1/3}\right),
\]
which is subsumed by the stated bound $\widetilde O\!\left(np^{1/3} + (np)^{2/3}\right)$, and the total running time is
\[
\widetilde O\!\left(p^{2/3}m + \sigma(m+p) + np^{1/3}\right).
\]

\subparagraph{Case 2: $p\ge n$.}
\begin{equation*} \beta_{\mathrm{I}} := n^{1/3}, \qquad \beta_{\mathrm{II}} := (n^2/p)^{1/3}. \end{equation*}
Since $P\subseteq S\times V\subseteq V\times V$, we have $p\le n^2$; together with $p\ge n$, this implies $1\le \beta_{\mathrm{II}}\le \beta_{\mathrm{I}}$.

\emph{Phase I.}
Construct a sequence $\beta_0=n > \beta_1 > \cdots > \beta_a = \beta_{\mathrm{I}}$ as follows: while the current value $\beta_{i-1}$ is larger than $2\beta_{\mathrm{I}}$, set
\begin{equation*} \beta_i := \beta_{i-1}^{2/3}\beta_{\mathrm{I}}^{1/3}. \end{equation*}
When the current value first becomes at most $2\beta_{\mathrm{I}}$, append one final level equal to $\beta_{\mathrm{I}}$.
Exactly as in Case~1, this gives $a=O(\log\log n)$. Every phase-I level lies in $[\beta_{\mathrm{I}},n]$, so \Cref{thm:hs-large-diam} and \Cref{lem:witness-paths-for-hopset-edges} give hopsets $H_1,\ldots,H_a$ with total construction time
\begin{equation*} \widetilde O\!\left((n/\beta_{\mathrm{I}})\cdot m\right) = \widetilde O\!\left(n^{2/3}m\right). \end{equation*}
Also, every phase-I level contributes
\begin{equation*} |H_i|\beta_{i-1} = \widetilde O\!\left(n^{3/2}/\beta_{\mathrm{I}}^{1/2}\right) = \widetilde O\!\left(n^{4/3}\right), \end{equation*}
so
\begin{equation*} \sum_{i=1}^a |H_i|\beta_{i-1} = \widetilde O\!\left(n^{4/3}\right). \end{equation*}

\emph{Phase II.}
Starting from $\gamma_0=\beta_{\mathrm{I}}$, construct a nonincreasing sequence $\gamma_0=\beta_{\mathrm{I}}\ge \gamma_1 > \cdots > \gamma_b = \beta_{\mathrm{II}}$ as follows: while the current value $\gamma_{j-1}$ is larger than $2\beta_{\mathrm{II}}$, set
\begin{equation*} \gamma_j := \gamma_{j-1}^{1/3}\beta_{\mathrm{II}}^{2/3}. \end{equation*}
When the current value first becomes at most $2\beta_{\mathrm{II}}$, append one final level equal to $\beta_{\mathrm{II}}$.
If $\gamma_0,\gamma_1,\ldots,\gamma_{b-1}$ are the values before the final appended level, then $y_j:=\gamma_j/\beta_{\mathrm{II}}$ satisfies $y_j=y_{j-1}^{1/3}$ for every $1\le j\le b-1$, with $y_0=\beta_{\mathrm{I}}/\beta_{\mathrm{II}}=(p/n)^{1/3}\le n^{1/3}$. Hence $b=O(\log\log n)$.
Rename the phase-II hopsets as
\[
H_{a+j}\text{ with hopbound }\beta_{a+j}:=\gamma_j
\qquad (j=1,\ldots,b),
\]
so that the full hierarchy is $H_1,\ldots,H_L$ with $L=a+b$ and final hopbound $\beta_L=\beta_{\mathrm{II}}$.

Every phase-II level lies in $[\beta_{\mathrm{II}},\beta_{\mathrm{I}}]\subseteq[1,n^{1/3}]$, so \Cref{thm:hs-small-diam} and \Cref{lem:witness-paths-for-hopset-edges} give hopsets with size $|H_{a+j}| = \widetilde O\!\left(n^2/\gamma_j^3\right)$ and construction time $\widetilde O(n/\gamma_j\cdot m)$. Since $\gamma_j^3=\gamma_{j-1}\beta_{\mathrm{II}}^2$ for every nonfinal level,
\[
|H_{a+j}|\beta_{a+j-1}
   = \widetilde O\!\left((n^2/\gamma_j^3)\gamma_{j-1}\right)
   = \widetilde O\!\left(n^2/\beta_{\mathrm{II}}^2\right)
   = \widetilde O\!\left((np)^{2/3}\right).
\]
The same estimate holds for the final level because $\gamma_b=\beta_{\mathrm{II}}$ and $\gamma_{b-1}=O(\beta_{\mathrm{II}})$. Therefore,
\begin{equation*} \sum_{j=1}^b |H_{a+j}|\beta_{a+j-1} = \widetilde O\!\left((np)^{2/3}\right). \end{equation*}
Similarly, the total construction time in phase II is
\begin{equation*} \widetilde O\!\left((n/\beta_{\mathrm{II}})\cdot m\right) = \widetilde O\!\left((np)^{1/3}m\right). \end{equation*}

Because the final hopset satisfies
\begin{equation*} |H_L| = \widetilde O\!\left(n^2/\beta_{\mathrm{II}}^3\right) = \widetilde O(p), \end{equation*}
the rooted retrieval stage again costs
\begin{equation*} \widetilde O\!\left(\sigma(m+p)\right), \end{equation*}
and the total length of all top-level source-to-target paths is
\begin{equation*} p\beta_L = p\beta_{\mathrm{II}} = (np)^{2/3}. \end{equation*}

\begin{equation*} \eta := \varepsilon/(2(L+1)). \end{equation*}
Then $(1+\eta)^{L+1}\le 1+\varepsilon$, so the  reduction is correct. Since $p\ge n$, we have $n^{4/3}\le (np)^{2/3}$, and therefore
\begin{equation*} \sum_{i=1}^L |H_i|\beta_{i-1} = \widetilde O\!\left((np)^{2/3}\right). \end{equation*}
Combining this with the  size bound yields
\[
|E(G_P)| = \widetilde O\!\left((np)^{2/3}\right),
\]
which is subsumed by the stated bound $\widetilde O\!\left(np^{1/3} + (np)^{2/3}\right)$. The total running time is
\[
\widetilde O\!\left(
(n/\beta_{\mathrm{I}})\cdot m + (n/\beta_{\mathrm{II}})\cdot m + \sigma(m+p) + p\beta_{\mathrm{II}}
\right)
=
\widetilde O\!\left(
(np)^{1/3}m + \sigma(m+p) + (np)^{2/3}
\right),
\]
because $(n/\beta_{\mathrm{I}})=n^{2/3}\le (np)^{1/3}$ when $p\ge n$.

This proves both running-time bounds and the stated edge bound.
\end{proof}

\section{Data Structure for Reporting Hopset Witnesses}

\label{app:hopset_data_structs}

In this section, we present the proofs of our hopset data structures that we use in the fast approximate distance preserver construction in \Cref{sec:fast-sourcewise-preservers}. We begin with \Cref{lem:rooted-path-retrieval}, which guarantees that we can compute hopbounded single-source (approximate) shortest paths efficiently, given our hopsets constructed in \Cref{thm:hs-large-diam} and \Cref{thm:hs-small-diam}.

\RootedPathRetrieval*

We want to find an efficient data structure to store the low-hop approximate shortest paths starting at $r$ guaranteed by the hopset.
The lemma shows this is possible by first computing a shortest path tree from $r$ and then rehanging parts of it using the chain shortcuts and the sampled-vertex-to-sampled-chain shortcuts from $H$. 
This reduces the height of the tree and thus the hop-length of the stored paths.
The properties of the hopset construction will guarantee that will not increase the weights of the paths beyond the desired approximation guarantee.

\begin{proof}
We build one structure per distance scale $L$ and use it for the pairs $(r,v)$ with $d_G(r,v)\in[L/(1+\eps_0),L]$. Let $V'$, $\mathcal C_L$, and $\mathcal R_L$ be the sampled vertices, the chains, and the sampled chains used to construct $H$ at this scale.
Let $T$ be the fixed shortest path tree from $r$ of \Cref{lem:prefix-sampling}, so its root paths are the paths $P_{r,v}$ and it does not depend on any random choice.
Every node $z$ stores a \emph{link}: an ancestor of $z$ in $T$ together with a path in $G\cup H$ from that ancestor to $z$. We obtain $\Pi_H(r,v)$ by following the links from $v$ up to $r$. By default, the link of $z$ is its tree edge. We change it in two cases, where (2) takes precedence over (1):
\begin{enumerate}
    \item If $z$ lies on a chain $Q\in\mathcal C_L$ and its highest ancestor $a$ on $Q$ is at least three hops above $z$, we link $z$ to $a$ via the exact two-hop path in $F_Q$.
    \item If $z$ lies on a sampled chain $Q\in\mathcal R_L$ and has depth more than $\beta$, let $u$ be the first node of $V'$ on its root path. We link $z$ to $u$ via the edge $(u,z')$ of \Cref{lem:earliest-bucket} (with $y=z$) and the exact two-hop path from $z'$ to $z$ in $F_Q$.
\end{enumerate}
By \Cref{lem:prefix-sampling}, w.h.p.\ such $u$ has depth at most $\beta$. Hence, at most one link of type (2) is used on $\Pi_H(r,v)$, all other links are exact, and exactly as in \Cref{lem:route-from-sampled-suffix}, we get $w(\Pi_H(r,v))\le(1+\delta)\,d_G(r,v)$.

For the hop bound, walk from $v$ upwards. By (1), we meet every chain in $\cO(1)$ nodes. We meet at most $2\beta$ nodes outside the chains: at most $\beta$ above the first node of $V'$ on $P_{r,v}$, and at most $\beta$ below it by \Cref{cond:d3}. Since $T$ and the links of type (1) do not depend on the sampling of the chains, w.h.p.\ one of the first $\beta$ distinct chains we meet is sampled.\footnote{This is for \Cref{thm:hs-small-diam}. For \Cref{thm:hs-large-diam}, only portals get links of types (1) and (2), and (1) points to the highest proper ancestor that is a portal of the same chain. Then following the links from $v$ is exactly the walk in the proof of \Cref{thm:hs-large-diam} for $P=P_{r,v}$, so its claim gives $\cO(\beta)$ hops.} At its node, we are either at depth at most $\beta$ or use a link of type (2) to jump to depth at most $\beta$. Hence, $\Pi_H(r,v)$ has $\cO(\beta)$ hops.

We compute $T$ by \Cref{lem:dag-sssp} and all links in one DFS over $T$, maintaining the highest node of every chain on the current root path and the first node of $V'$ on it, and undoing changes when backtracking. The edge $(u,z')$ is found by binary search over the $\cO(\eps_0^{-1}\log(nW))$ bucket edges from $u$ to $Q$, and the two-hop paths in $F_Q$ via its recursive construction. Hence, the total time is $\Oish(m+|H|)$.
\end{proof}

At a high level,  \Cref{lem:witness-paths-for-hopset-edges} guarantees that when queried with an edge from a hopset $H$ constructed using \Cref{thm:hs-large-diam} or \Cref{thm:hs-small-diam}, we can recover a low-hop path in the underlying graph between the endpoints of that hopset edge.

\WitnessPathsForHopsetEdges*

\begin{proof}
We assume our hopset hierarchy was constructed using \Cref{thm:hs-small-diam} and \Cref{thm:hs-large-diam}. Moreover, assume that our data structure has access to the sampled vertices and (sampled) chains used in the construction of these hopsets.

We divide the edges in hopset $H_i$ into two categories: edges added while super-shortcutting chains and edges added from sampled vertices (lets say $V'_i$) to sampled chains.
Call these \emph{type one} edges and \emph{type two} edges. We treat them differently.

For type one edges, recall how super-shortcutting a chain (for \Cref{thm:hs-large-diam}, the chain $\Pi_Q$ of portals) works: We find the middle vertex of the chain and add edges from all predecessors on the chain to it and add edges to all its descendants of the chain. Then, recurse on the half of the chain before the middle vertex and the half after separately.
We mirror this recursion and apply \Cref{lem:rooted-path-retrieval} with the middle node as $r$ with hopset $H_{i-1}$.
Afterwards, we split the instance into two halves $G \cup H_{i-1}[<r]$ and $G \cup H_{i-1}[>r]$.
Notice two things: (1) as we operate on DAGs, the only part of the graph that is necessary to find the witnessing path for an edge $(u,v)$ are the vertices between $u,v$ in the topological partial order, and (2) the sum of the sizes of these two instances is at most $\cO(m)$ (after a one-time preprocessing of $H_{i-1}$, \Cref{lem:rooted-path-retrieval} only scans the $G$-edges of an instance).
Thus, constructing the trees from \Cref{lem:rooted-path-retrieval} required to answer queries for edges of type one, takes $\Oish(m)$ time per chain used to construct $H_i$, thus $\Oish(n/\beta_i \cdot m)$ in total.

For type two, it is directly enough to run \Cref{lem:rooted-path-retrieval} with all vertices from $V'_i$ as roots.
This takes $\Oish(n/\beta_i \cdot m)$ time.
\end{proof}

\section{Reductions to DAGs due to~\cite{haeuplerJiangSaranurak26}}

\label{app:dagreduction}

In this section, we state the two lemmas due to \cite{haeuplerJiangSaranurak26} that allow us to assume our input graphs are DAGs in our fast hopset and distance preserver algorithms.

\begin{lemma}[\cite{haeuplerJiangSaranurak26}]
    Suppose there exists an algorithm that, given an $m$-edge, $n$-vertex DAG, computes a $(1+\varepsilon)$-approximate hopset of size $s(m, n)$ and hopbound $\beta$ in time $T(m, n)$, for $1/\polylog(n) < \varepsilon < o(1)$. Then there exists a randomized algorithm that, given an $m$-edge, $n$-vertex directed graph, computes a $(1+3\varepsilon)$-approximate hopset of size $s(m^{1+o(1)}, n^{1+o(1)})$ and hopbound $\beta$ in time $T(m^{1+o(1)}, n^{1+o(1)}) + m^{1+o(1)}$. 
\end{lemma}

\begin{lemma}[\cite{haeuplerJiangSaranurak26}]
    Suppose there exists an algorithm that, given an $m$-edge, $n$-vertex DAG and a set of $p$ demand pairs with $\sigma$ sources, computes a $(1+\varepsilon)$-approximate distance preserver of size $s(n, p)$ in time $T(m, n, p, \sigma)$, for $1/\polylog(n) < \varepsilon < o(1)$. Then there exists a randomized algorithm that, given an $m$-edge, $n$-vertex directed graph, computes a $(1+3\varepsilon)$-approximate distance preserver of size $s(n^{1+o(1)}, p)$ in time $T(m^{1+o(1)}, n^{1+o(1)}, p, \sigma) + m^{1+o(1)}$. 
\end{lemma}

\section{A Simple Algorithm for $(1+\varepsilon)$-Hopsets on DAGs due to~\cite{KP22a}}
\label{sec:kp-hopset}

In this section, we present a simple algorithm for $(1+\varepsilon)$-hopsets on DAGs that is implicit in the work of \cite{KP22a}. The reason we present this algorithm is because it is the basis for our faster algorithms for $(1+\varepsilon)$-hopsets.

\begin{theorem}
Let $G=(V,E,w)$ be an $n$-vertex DAG with nonnegative edge-weights and polynomial aspect ratio $W=\poly(n)$. For every $\eps\in(0,1)$ there exists a $(1+\eps)$-hopset $H$ of size $\Oish(\eps^{-2} n)$ and hopbound $\Oish(\eps^{-1} n^{1/3})$.
\label{thm:easy_dag_hopset}
\end{theorem}

Let $\tc(G)$ denote the metric closure of $G$: for every reachable pair $(u,v)$, $\tc(G)$ contains an arc $(u,v)$ of weight $d_G(u,v)$. A \emph{chain} is an arbitrary directed path in $\tc(G)$. We say that a chain is a \emph{lightest chain} if it is a shortest path in $\tc(G)$.

\begin{definition}[Lightest chain]
A chain $Q=(q_1,q_2,\dots,q_h)$ in $\tc(G)$ is a \emph{lightest chain} if
\[
w(Q)=\sum_{i=1}^{h-1} d_G(q_i,q_{i+1}) = d_G(q_1,q_h).
\]
\end{definition}

\subparagraph{Sampling vertices.}
Fix a parameter $\beta\ge 1$, which in the proof of \Cref{thm:easy_dag_hopset} we will later set to
\[
\beta:=\ceil{n^{1/3}}.
\]
Sample each vertex in $V(G)$ independently with probability
\[
p:=\min\!\left\{1,\frac{C\log n}{\beta}\right\},
\]
where $C$ is a sufficiently large absolute constant, and let $V'\subseteq V$ be the sampled set. For every $s$, fix a shortest path tree from $s$ independently of the sampling (e.g., by \Cref{lem:dag-sssp} with consistent tie-breaking), and let $P_{s,t}$ be its $s$-$t$ path. Our sampled vertices provide the following guarantee.

\begin{lemma}[Vertex sampling guarantee]
\label{lem:prefix-sampling}
With high probability, the following two events hold simultaneously:
\begin{enumerate}
    \item $|V'|=\Oish(n/\beta)$.
    \item For every reachable ordered pair $(s,t) \in \tc(G)$, if $P \coloneqq P_{s,t}$ has $|V(P)| > \beta$, then one of the first $\beta$ vertices of $P$ is contained in $V'$.
\end{enumerate}
\end{lemma}

\begin{proof}
The bound on $|V'|$ follows from a standard Chernoff bound. For the second part, fix a reachable pair $(s,t)$ such that $P = P_{s,t}$ has more than $\beta$ vertices. If $p=1$, then the claim is immediate. Otherwise, the probability that none of the first $\beta$ vertices of $P$ is sampled is at most
\[
(1-p)^\beta \le e^{-p\beta}=n^{-C}.
\]
A union bound over all ordered pairs in $\tc(G)$ finishes the proof (this is why we fix one path per pair: there can be exponentially many shortest paths).
\end{proof}

We now build the hopset as a union over distance scales.
Let $W_{\min} := \min_{e \in \tc(G)} w(e)$ and $W_{\max} := \max_{e \in \tc(G)} w(e)$.\footnote{We assume positive edge weights. Zero weights can be handled by adding $\eps w^{+}/(4n)$ to every edge weight, where $w^{+}$ is the smallest positive weight, and computing the hopset of the resulting graph, with its edges weighted by their distances in $G$.} Let
\[
\lset=\set{W_{\min},(1+\eps)W_{\min},(1+\eps)^2W_{\min},\dots, W_{\max}}
\]
be a bucketing over all finite distances in $[W_{\min},W_{\max}]$, and note that
\[
|\lset| = \cO(\varepsilon^{-1}\log(Wn)),
\]
where $W$ is the aspect ratio of graph $G$. For each $L\in\lset$ we will construct a hopset $H_L$ that handles all ordered pairs $(s,t) \in \tc(G)$ with
\[
d_G(s,t)\in [L/(1+\eps),L].
\]
Since $|\lset|=\cO(\eps^{-1}\log(Wn))$, it will be sufficient to prove $|H_L|=\Oish(\eps^{-1}n)$ for one fixed scale.

For the remainder of this section we fix a scale $L \in \lset$ and focus on constructing $H_L$.

\begin{lemma}[Maximal lightest chain cover at scale $L$]
\label{lem:maximal-cover}
There exists a family $\qset_L$ of pairwise vertex-disjoint lightest chains in $\tc(G)$ such that every $Q\in\qset_L$ has at least $\beta$ vertices and weight at most $\varepsilon L$, and the family is maximal under inclusion. Consequently,
\[
|\qset_L|\le n/\beta.
\]
Moreover, if $P$ is a lightest chain with $w(P)\le L$, then at most $O(\varepsilon^{-1}\beta)$ vertices of $P$ lie outside $\bigcup_{Q\in\qset_L} V(Q)$.
\end{lemma}

\begin{proof}
Start with the empty family and keep adding eligible chains until no more can be added. This gives the maximal family $\qset_L$. Since the lightest chains are vertex-disjoint and each contains at least $\beta$ vertices, the size bound follows immediately.

Now let $P$ be a lightest chain with $w(P)\le L$. Partition $P$ into at most $\ceil{1/\varepsilon}$ consecutive subchains, each of weight at most $\varepsilon L$. If any such subchain contained at least $\beta$ vertices that do not lie in $\bigcup_{Q\in\qset_L} V(Q)$, then those uncovered vertices, in their order on $P$, would form a new lightest chain in $\tc(G)$ with at least $\beta$ vertices and weight at most $\varepsilon L$. It would also be vertex-disjoint from every lightest chain already in $\qset_L$, contradicting maximality. So each subchain contains fewer than $\beta$ uncovered vertices, and the total number of uncovered vertices on $P$ is at most $\ceil{1/\varepsilon}\beta$.
\end{proof}

\begin{lemma}[Shortcutting a chain]
\label{lem:shortcut-one-chain}
Let $Q=(q_1,q_2,\dots,q_h)$ be a lightest chain in $G$. There is a set $F_Q \subseteq \tc(G)$ of $|F_Q| = \cO(h\log h)$ edges such that for every $1\le i<j\le h$, $F_Q$ contains a $q_i\leadsto q_j$ path of weight exactly $d_G(q_i,q_j)$ and with at most two hops.
\end{lemma}

\begin{proof}
This lemma can be proven using a recursive shortcutting procedure (see, e.g., \cite{Raskhodnikova10} for details).
\end{proof}

Add the edges of $F_Q$ to $H_L$ for every chain $Q\in\qset_L$. Since the chains in $\qset_L$ are vertex-disjoint, these exact chain shortcuts contribute only $\cO(n \log n)$ edges to hopset $H_L$.

\subparagraph{Sampling chains.}
Independently sample each chain in $\qset_L$ with the same probability
\[
p:=\min\!\left\{1,\frac{C\log n}{\beta}\right\},
\]
where $C>0$ is a sufficiently large constant, and let $\rset_L\subseteq \qset_L$ be the sampled subfamily.

\begin{lemma}[Sampled chains]
\label{lem:sampled-chains}
With high probability,
\[
|\rset_L|=\Oish\!\left(|\qset_L|\cdot \min\!\left\{1,\frac{\log n}{\beta}\right\}\right).
\]
In particular, if $|\qset_L|\le n/\beta$, then
\[
|\rset_L|=\Oish(n/\beta^2).
\]
\end{lemma}

\begin{proof}
This is another standard Chernoff bound.
\end{proof}

For every $v\in V'$ and every $Q\in\rset_L$, we add a small set of forward edges from $\{v\} \times Q$ to our hopset $H_L$. Let
\[
J:=\left\lceil \log_{1+\varepsilon}(W_{\max}/W_{\min}) \right\rceil.
\]
For every $j\in\{0,1,\dots,J\}$ define
\[
U_j(v,Q):=\set{ x\in V(Q) : d_G(v,x) \le (1+\varepsilon)^{j} W_{\min} }.
\]
If $U_j(v,Q)\neq\emptyset$, let $x_j(v,Q)$ be the earliest vertex of $U_j(v,Q)$ on $Q$, and add the edge $(v,x_j(v,Q))$ with weight $d_G(v,x_j(v,Q))$ to hopset $H_L$. We define our final hopset $H$ of $G$ to be $H := \cup_{L \in \lset}H_L$.

\begin{lemma}[Earliest-bucket edges]
\label{lem:earliest-bucket}
Fix $v\in V'$, $Q\in \rset_L$, and $y\in V(Q)$ reachable from $v$. Then $H_L$ contains an edge $(v,x)$ such that $x$ is not later than $y$ on $Q$ and
\[
d_G(v,x)\le (1+\eps)\,d_G(v,y).
\]
\end{lemma}

\begin{proof}
Let $j$ be the unique index with
\[
(1+\eps)^{j-1}W_{\min}<d_G(v,y)\le (1+\eps)^jW_{\min}.
\]
Then $y\in U_j(v,Q)$, so $x_j(v,Q)$ is defined. By construction it is the earliest vertex of $U_j(v,Q)$ on $Q$, hence it is not later than $y$ on $Q$. Also,
\[
d_G(v,x_j(v,Q))\le (1+\eps)^jW_{\min}\le (1+\eps)d_G(v,y).
\]
\end{proof}

\begin{lemma}[Rerouting through chain shortcuts]
\label{lem:reroute-through-chains}
Let $P$ be a shortest path in $G$ (viewed as a lightest chain in $\tc(G)$), and let $\mathcal Q$ be a family of pairwise vertex-disjoint lightest chains. After adding the shortcut sets $F_Q$ from \Cref{lem:shortcut-one-chain} for all $Q\in\mathcal Q$, there is a path $P'$ in $G\cup \bigcup_{Q\in\mathcal Q}F_Q$ with the same endpoints and the same weight as $P$ such that every chain $Q\in\mathcal Q$ contributes at most three vertices to $P'$. Moreover, every vertex of $P'$ outside $\bigcup_{Q\in\mathcal Q}V(Q)$ already belongs to $P$ outside that union.
\end{lemma}

\begin{proof}
We process the chains of $\mathcal Q$ one by one. At every step the current path is still a lightest chain in $\tc(G)$, because all the edges we use have weight equal to the original graph distance between their endpoints.

Fix a chain $Q\in\mathcal Q$. If the current path meets $Q$ in at most three vertices, do nothing. Otherwise, let $a$ and $b$ be the first and last vertices of the current path that lie on $Q$. Since $G$ is acyclic, $a$ and $b$ appear in the same order on the current path and on $Q$. By \Cref{lem:shortcut-one-chain}, $F_Q$ contains an $a\leadsto b$ path of at most two hops and weight exactly $d_G(a,b)$. Replace the whole subpath between $a$ and $b$ by this exact shortcut. This preserves the total weight, and afterwards the current path uses at most three vertices of $Q$.

Because the chains in $\mathcal Q$ are vertex-disjoint, later replacements never touch the vertices of $Q$. Also, every replacement only shortcuts through vertices of the processed chain, so no new vertex outside $\bigcup_{Q\in\mathcal Q}V(Q)$ is introduced. After all chains have been processed, the resulting path is the desired path $P'$.
\end{proof}

\begin{lemma}[Routing from a sampled suffix]
\label{lem:route-from-sampled-suffix}
Suppose every chain $Q\in\qset_L$ has weight at most $\eta L$, where $\eta>0$. Let $P$ be a shortest $s\leadsto t$-path in $G$ with
\[
d_G(s,t)\in [L/(1+\eps),L].
\]
Suppose some vertex $v_1\in V'$ lies among the first $\beta$ vertices of $P$, and the suffix $P[v_1,t]$ contains at most $\kappa$ vertices outside $\bigcup_{Q\in\qset_L}V(Q)$, where $\kappa\ge \beta$. Then, with high probability, $G\cup H_L$ contains an $s\leadsto t$-path $\widehat P$ such that
\[
|\widehat P|=O(\kappa)
\qquad\text{and}\qquad
w(\widehat P)\le \bigl(1+\eps+\eta(1+\eps)\bigr)d_G(s,t).
\]
\end{lemma}

\begin{proof}
Intuitively, we follow $P$ exactly up to $v_1$, jump from $v_1$ to a sampled chain near the end of $P$, and follow the (rerouted) path from there; only the jump loses a factor $1+\eps$. Apply \Cref{lem:reroute-through-chains} to the suffix $P[v_1,t]$ and the family $\qset_L$. We obtain a path $P'$ from $v_1$ to $t$ of the same weight such that every chain $Q\in\qset_L$ contributes at most three vertices to $P'$, and $P'$ still contains at most $\kappa$ vertices outside $\bigcup_{Q\in\qset_L}V(Q)$.

If $|P'|=\cO(\kappa)$, then the concatenation $P[s,v_1]\circ P'$ already has $\cO(\kappa)$ hops and exact weight $d_G(s,t)$, since $\kappa\ge \beta$.

Otherwise, let $P_2$ be the suffix of $P'$ consisting of its last $C_0\kappa$ vertices, for a sufficiently large absolute constant $C_0$. Since $P'$ has at most $\kappa$ vertices outside the chains, at least $(C_0-1)\kappa$ vertices of $P_2$ lie on chains of $\qset_L$. Because every chain contributes at most three vertices to $P'$, the suffix $P_2$ intersects $\Omega(\kappa)$ distinct chains of $\qset_L$.

By the chain-sampling probability $p=\min\{1,C\log n/\beta\}$ and the assumption $\kappa\ge \beta$, with high probability at least one of these chains belongs to $\rset_L$. Fix such a sampled chain $Q$, and let $v_2\in V(Q)\cap V(P_2)$. By \Cref{lem:earliest-bucket}, there is an edge $(v_1,x)\in H_L$ such that $x$ is not later than $v_2$ on $Q$ and
\[
d_G(v_1,x)\le (1+\varepsilon)d_G(v_1,v_2).
\]
Also, by \Cref{lem:shortcut-one-chain}, there is an $x\leadsto v_2$ path inside $F_Q$ of at most two hops and weight exactly $d_G(x,v_2)$. Since $Q$ is a lightest chain of weight at most $\eta L$,
\[
d_G(x,v_2)\le w(Q)\le \eta L \le \eta(1+\eps)d_G(s,t).
\]

Consider the path $\widehat P$ that follows $P[s,v_1]$, then the hopset edge $(v_1,x)$, then the at-most-two-hop exact route from $x$ to $v_2$ inside $Q$, and finally the suffix $P'[v_2,t]$. Its hop count is
\[
|\widehat P|
\le |P[s,v_1]|+1+2+|P'[v_2,t]|
=
\cO(\kappa),
\]
because $|P[s,v_1]|\le \beta\le \kappa$ and $v_2$ lies among the last $C_0\kappa$ vertices of $P'$.

For the weight, note that every edge of $P'$ has weight equal to the original distance between its endpoints, so the concatenation $P[s,v_1]\circ P'$ is a lightest chain from $s$ to $t$. Since $v_2$ appears after $v_1$ on this chain,
\[
w(P[s,v_1]) + d_G(v_1,v_2) + w(P'[v_2,t]) = d_G(s,t).
\]
Therefore
\begin{align*}
w(\widehat P)
&= w(P[s,v_1]) + d_G(v_1,x) + d_G(x,v_2) + w(P'[v_2,t])\\
&\le w(P[s,v_1]) + (1+\varepsilon)d_G(v_1,v_2) + d_G(x,v_2) + w(P'[v_2,t])\\
&= d_G(s,t) + \varepsilon d_G(v_1,v_2) + d_G(x,v_2)\\
&\le d_G(s,t) + \varepsilon d_G(s,t) + \eta(1+\eps)d_G(s,t).
\end{align*}
This proves the claimed stretch bound.
\end{proof}

\begin{lemma}[Constructing one scale once the chains are known]
\label{lem:build-one-scale}
Once the set $V'$ and a pairwise vertex-disjoint family $\qset_L$ are fixed, the corresponding edge set $H_L$ can be generated in $\Oish(|V'|\cdot m + |H_L|)$ time.
\end{lemma}

\begin{proof}
For every chain $Q\in\qset_L$, the set $F_Q$ from \Cref{lem:shortcut-one-chain} can be generated in time proportional to its size. Since the chains are vertex-disjoint, this costs $\Oish(n)$ in total.

For the bucket edges, run one shortest-path computation from each source vertex $v\in V'$. For each such source we then scan the sampled chains in their internal order and emit the earliest edge for every nonempty bucket. Because the chains are vertex-disjoint, this scan costs $\Oish(n)$ per source vertex, and the total time to emit the bucket edges is additionally proportional to their number. This gives total running time $\Oish(|V'|\cdot m + |H_L|)$.
\end{proof}

\begin{proof}[Proof of \Cref{thm:easy_dag_hopset}]
Set
\(
\beta:=\ceil{n^{1/3}}.
\) For every scale $L\in\lset$, construct the maximal family $\qset_L$ from \Cref{lem:maximal-cover}, add the exact shortcut sets $F_Q$ for all $Q\in\qset_L$, sample the subfamily $\rset_L$, and add the corresponding bucket edges.

Let us first bound the size. By \Cref{lem:prefix-sampling},
\[
|V'|=\Oish(n/\beta).
\]
By \Cref{lem:maximal-cover} and \Cref{lem:sampled-chains},
\[
|\rset_L|=\Oish(n/\beta^2).
\]
Thus the bucket edges contribute
\[
|V'|\cdot |\rset_L|\cdot \cO(\eps^{-1}\log(Wn))
=
\Oish(\eps^{-1}n^2/\beta^3)
=
\Oish(\eps^{-1}n)
\]
edges at scale $L$. The exact chain shortcuts contribute only $\Oish(n)$ more edges at that scale. Since $|\lset|=\cO(\eps^{-1}\log(Wn))$, altogether
\(
|H|=\Oish(\eps^{-2}n).
\)

Now fix a reachable ordered pair $(s,t)$, and let $P \coloneqq P_{s,t}$, viewed as a lightest chain in $\tc(G)$. Let $L\in\lset$ be the unique scale with
\[
d_G(s,t)\in [L/(1+\eps),L].
\]
If $|P|\le \beta$, then there is nothing to prove. Otherwise, \Cref{lem:prefix-sampling} gives a vertex $v_1\in V'$ among the first $\beta$ vertices of $P$. Since
\[
w(P[v_1,t])\le d_G(s,t)\le L,
\]
\Cref{lem:maximal-cover} implies that the suffix $P[v_1,t]$ contains at most
\(
\kappa:=\cO(\beta/\eps)
\)
vertices outside $\bigcup_{Q\in\qset_L}V(Q)$. Applying \Cref{lem:route-from-sampled-suffix} with $\eta=\eps$ and $\kappa = \cO(\beta/\eps)$ yields, with high probability, an $s\leadsto t$-path in $G\cup H_L$ with
\(
\cO(\beta/\eps)
\)
hops and stretch $1+\cO(\eps)$; rescaling $\eps$ by a constant gives the claim.
\end{proof}

\section{Deriving Fast Hopsets For All Parameters}
\label{sec:fast-all-params}

The construction for $(1+\varepsilon)$-hopsets in \Cref{sec:kp-hopset} (and in particular, \Cref{lem:build-one-scale}) suggests that to compute $(1+\varepsilon)$-hopsets in our desired time complexity, it is sufficient to compute a lightest chain cover quickly. In this section, we formally show that plugging in the light chain cover construction described in \Cref{alg:lightest-chain-cover} into the hopset construction described in \Cref{sec:kp-hopset} yields the following lemma.

\begin{lemma}
\label{thm:hs-small-diam}
    Let $G=(V,E,w)$ be an $n$-vertex, $m$-edge directed acyclic graph with nonnegative edge-weights and polynomial aspect ratio. For any $\varepsilon \in (0, 1)$ and hopbound $\beta = \BigO(n^{1/3})$, there is an
$\Oish_{\varepsilon}\!\left(\frac{n}{\beta}\cdot m \right)$-time algorithm computing
$\Oish_{\varepsilon}\!\left(\frac{n^2}{\beta^3}\right)$-size $(1+\varepsilon)$-hopsets with $\BigO(\beta)$ hopbound.
\end{lemma}

Later in this section, in \Cref{thm:hs-large-diam}, we will give a separate proof of sublinear-size $(1+\varepsilon)$-hopsets for the case where $\beta \in \omega(n^{1/3})$. This  needs an additional subsampling trick and randomly chosen portals, which is  why we handle it separately. This idea is inspired by a similar trick from \cite{kogan_faster_nodate}, originally developed for shortcut sets.

\begin{proof}[Proof of \Cref{thm:hs-small-diam}]
We follow the construction from \Cref{sec:kp-hopset}, and only replace the chain-cover step by \Cref{lem:fast_chain_alg}. Set $\varepsilon_0:=\varepsilon/20$ and $\eta:=\varepsilon_0/10$, and let \begin{equation*}
\mathcal L:=\{W_{\min},(1+\varepsilon_0)W_{\min},(1+\varepsilon_0)^2W_{\min},\ldots,W_{\max}\}, \qquad H:=\bigcup_{L\in\mathcal L} H_L.
\end{equation*}
For each $L\in\mathcal L$, the scale-$L$ hopset handles all pairs $(s,t)$ with \begin{equation*}d_G(s,t)\in [L/(1+\varepsilon_0),L].\end{equation*}

Fix a scale $L$. Sample every vertex with probability $p:=\min\{1,C\log n/\beta\}$ and let $V'$ be the sampled set, where $C$ is a large enough constant again. By \Cref{lem:prefix-sampling}, with high probability
\begin{equation*}|V'|=\Oish(n/\beta).\end{equation*}
Apply \Cref{lem:fast_chain_alg} directly to $G$ and this set $V'$ to obtain a family $\mathcal C_L$ satisfying \Cref{def:lightest-chain-cover}. Thus \begin{equation*}|\mathcal C_L|=\Oish_\varepsilon(n/\beta),\end{equation*} every $Q\in\mathcal C_L$ is a lightest chain of weight at most $\eta L$, and the chains are pairwise vertex-disjoint.

Now build $H_L$ exactly as in \Cref{sec:kp-hopset}: add the exact shortcut sets $F_Q$, sample a subfamily $\mathcal R_L$, and add the bucket edges from $V'$ to $\mathcal R_L$. Using \Cref{lem:sampled-chains} and the disjointness of the chains,
\begin{equation*}|\mathcal R_L|=\Oish_\varepsilon(n/\beta^2), \qquad |H_L|=\Oish_\varepsilon(n^2/\beta^3).\end{equation*}
Since $|\mathcal L|=\Oish_\varepsilon(1)$, this gives \begin{equation*}|H|=\Oish_\varepsilon(n^2/\beta^3).\end{equation*}
By \Cref{lem:fast_chain_alg,lem:build-one-scale}, the total running time is
\begin{equation*}\Oish_\varepsilon\!\left(\frac{n}{\beta}\cdot m + \frac{n^2}{\beta^3}\right) = \Oish_\varepsilon\!\left(\frac{n}{\beta}\cdot m\right),\end{equation*}
where we assume w.l.o.g.\ that $n \le 2m$, as isolated vertices can be ignored.

For stretch and hopbound, fix a reachable pair $(s,t)$ and let $P$ be the fixed shortest $s\leadsto t$-path $P_{s,t}$ of \Cref{lem:prefix-sampling}, viewed as a lightest chain in $\tc(G)$. Choose the unique scale $L\in\mathcal L$ with \begin{equation*}d_G(s,t)\in [L/(1+\varepsilon_0),L].\end{equation*}
If $|P|\le \beta$, then $P$ itself already has at most $\beta$ hops. Otherwise, by \Cref{lem:prefix-sampling}, some vertex $v\in V'$ lies among the first $\beta$ vertices of $P$. Since \begin{equation*}w(P[v,t])=d_G(v,t)\le d_G(s,t)\le L,\end{equation*} the cover property of $\mathcal C_L$ implies that the suffix $P[v,t]$ contains at most $\beta$ vertices outside $\bigcup_{Q\in\mathcal C_L}V(Q)$. Applying \Cref{lem:route-from-sampled-suffix} with $\kappa=\beta$, the family $\mathcal C_L$, and parameter $\varepsilon_0$ yields an $s\leadsto t$-path in $G\cup H_L$ with $\cO(\beta)$ hops and weight at most
\begin{equation*}\bigl(1+\varepsilon_0+\eta(1+\varepsilon_0)\bigr)d_G(s,t)\le (1+\varepsilon)d_G(s,t).\end{equation*}

A union bound over all scales and ordered pairs shows that, with high probability, one random choice of the samples works simultaneously for every pair. Hence the algorithm computes a $(1+\varepsilon)$-hopset with the claimed size, hopbound, and running time.
\end{proof}

We also obtain a faster construction of $(1+\varepsilon)$-hopsets with $o(n)$ edges and hopbound $\beta  \in [n^{1/3}, n]$, by a slight modification of  the construction in \Cref{sec:kp-hopset}.

\begin{restatable}{lemma}{hslargediamthm}\label{thm:hs-large-diam}
Let $G=(V,E,w)$ be an $n$-vertex DAG with $m$ edges, nonnegative edge-weights, and polynomial
aspect ratio. Let $\varepsilon\in(0,1)$, and let $\beta$ be an integer with $n^{1/3} \le \beta \le n$.
Then one can compute a $(1+\varepsilon)$-hopset $H$ for $G$ with hopbound at most $O(\beta)$ and size at most $|H|=\widetilde O_{\varepsilon}\!\left(\left(\frac{n}{\beta}\right)^{3/2}\right)$
in running time
$
\Oish_{\varepsilon}(n/\beta \cdot m).
$
\end{restatable}
\begin{proof}
As in \Cref{sec:kp-hopset}, we work one distance scale at a time and then take the union over all scales.

Set $\eps_0:=\eps/20$, and
$\eta:=\eps_0/10$. Let
\begin{equation*}
\mathcal L:=\{W_{\min},(1+\eps_0)W_{\min},(1+\eps_0)^2W_{\min},\ldots,W_{\max}\},
\qquad
H:=\bigcup_{L\in\mathcal L} H_L.
\end{equation*}
Since the aspect ratio is polynomial, $|\mathcal L|=\Oish_\eps(1)$.
For each $L\in\mathcal L$, the hopset $H_L$ will handle all reachable pairs $(s,t)$ with 
$d_G(s,t)\in [L/(1+\eps_0),L]$.

Fix one scale $L$, and write
\begin{equation*}
r:=\left\lceil \sqrt{n/\beta}\right\rceil,
\qquad
\delta:=\left\lceil \beta/r\right\rceil.
\end{equation*}
Since $\beta \ge n^{1/3}$, we have $r\le \beta$, and therefore $1\le \delta\le \beta$.
Sample every vertex in $V(G)$ independently with probability
$
p_v:=\min\!\left\{1,\frac{C\log n}{\beta}\right\},
$ 
and let $V'$ be the sampled set.
By \Cref{lem:prefix-sampling}, with high probability, $
|V'|=\Oish(n/\beta)$,
and every path $P_{s,t}$ with more than $\beta$ vertices contains some vertex of $V'$ among its first $\beta$ vertices.

Apply \Cref{lem:fast_chain_alg} with parameter $\beta$ to obtain a family $\mathcal C_L$ satisfying \Cref{def:lightest-chain-cover}. Thus $|\mathcal C_L|=\Oish(n/\beta)$, 
every $Q\in\mathcal C_L$ is a lightest chain of weight at most $\eta L$, and the chains are pairwise vertex-disjoint.

\subparagraph{Sparse exact hop edges.}
Independently of all other random choices, make every vertex a \emph{portal} with probability $\rho:=\min\{1,\,C r\log n/\beta\}$, so $\rho\approx C\log n/\delta$. For each chain $Q\in\mathcal C_L$, the portals of $Q$ form a subsequence $\Pi_Q$, which is again a lightest chain. Apply \Cref{lem:shortcut-one-chain} to $\Pi_Q$, and let $F_Q^{\delta}$ be the resulting edge set. So any two portals of $Q$ are joined by an exact path with at most two hops. We cannot walk along $Q$ between other vertices, since consecutive vertices of $Q$ are in general only joined by edges of $\tc(G)$. Instead, the analysis below follows the shortest path itself and only jumps between portals that this path visits. Random portals ensure that it visits them often enough.

Since the chains are vertex-disjoint and, by Chernoff, w.h.p.\ there are $\Oish(\rho n)$ portals,
\begin{equation*}
\sum_{Q\in\mathcal C_L}|F_Q^{\delta}|
=\Oish\!\left(\frac{n}{\delta}\right)
=\Oish\!\left(\frac{nr}{\beta}\right)
=\Oish\!\left(\left(\frac{n}{\beta}\right)^{3/2}\right).
\end{equation*}

\subparagraph{Subsampling the chains.}
Sample each chain in $\mathcal C_L$ independently with probability
\begin{equation*}
q:=\min\!\left\{1,\frac{C\log n}{r}\right\},
\end{equation*}
and let $\mathcal R_L\subseteq \mathcal C_L$ be the sampled subfamily. By Chernoff, or trivially when $q=1$,
\begin{equation*}
|\mathcal R_L|=\Oish\!\left(\frac{|\mathcal C_L|}{r}\right)=\Oish\!\left(\sqrt{\frac{n}{\beta}}\right).
\end{equation*}
For every $v\in V'$ and every $Q\in\mathcal R_L$, add the earliest-bucket edges from $v$ to $\Pi_Q$, i.e., to the earliest portal of $Q$ in each bucket, exactly as in \Cref{sec:kp-hopset}. Let $H_L$ be the union of all edges in the sets $F_Q^{\delta}$ and all these bucket edges.

Using the same counting as in \Cref{sec:kp-hopset},
\begin{equation*}
|H_L|
=\Oish_\eps\!\left(\frac{n}{\delta}+|V'|\cdot |\mathcal R_L|\right)
=\Oish_\eps\!\left(\left(\frac{n}{\beta}\right)^{3/2}\right).
\end{equation*}
By \Cref{lem:fast_chain_alg} and the same implementation as in \Cref{lem:build-one-scale}, the scale-$L$ construction takes time
\begin{equation*}
\Oish_\eps\!\left(\frac{n}{\beta}\cdot m + |H_L|\right)=\Oish_\eps\!\left(\frac{n}{\beta}\cdot m\right),
\end{equation*}
because $|H_L|=\Oish_\eps((n/\beta)^{3/2})\le \Oish(n)$.
Summed over all scales, this is $\Oish_\eps(n/\beta\cdot m)$.

\subparagraph{Correctness.}
Fix a reachable ordered pair $(s,t)$, and let $P$ be the fixed shortest path $P_{s,t}$. Choose the unique scale $L\in\mathcal L$ for which
\begin{equation*}
d_G(s,t)\in [L/(1+\eps_0),L].
\end{equation*}
If $|P| \le 10\beta$, there is nothing to prove.
Assume therefore that $|P|> 10\beta$.
By \Cref{lem:prefix-sampling}, some vertex $v_1\in V'$ lies among the first $\beta$ vertices of $P$.
Since $w(P[v_1,t])\le L$, property~\ref{cond:d3} of \Cref{def:lightest-chain-cover} implies that the suffix $P[v_1,t]$ contains at most $\beta$ vertices outside $\bigcup_{Q\in\mathcal C_L}V(Q)$.

Consider the walk from $t$ backwards along $P$ that, whenever it is at a portal of a chain $Q$ and $P$ contains an earlier portal of $Q$, jumps to the first portal of $Q$ on $P$ using the two exact hops of $F_Q^{\delta}$. Read forwards, this is an exact path.

We claim that w.h.p., within $10\beta$ hops, the walk reaches $P[s,v_1]$ or visits portals of $r$ distinct chains. In the first case we are done, as the rest of $P$ has at most $\beta$ hops. For the claim, suppose that neither happens. After the walk has visited a portal of a chain $Q$, it jumps at most once and then only visits vertices of $Q$ on $P$ before its first portal. By the choice of $\rho$, w.h.p.\ there are at most $\beta/(2r)$ of them. Together with the at most $\beta$ vertices outside the chains, this accounts for at most $\beta+r(2+\beta/(2r))\le 4\beta$ of the first $10\beta$ hops. Each of the other hops reaches a vertex of a chain without a visited portal, which is a portal with probability $\rho$, independently of the walk so far. So by the choice of $\rho$, w.h.p.\ at least $r$ of these hops reach a portal, a contradiction.

Since the chains are sampled independently of the walk, w.h.p.\ one of these $r$ chains, say $Q$, belongs to $\mathcal R_L$. Let $v_2$ be the first portal of $Q$ that the walk visits; it lies on $P$ after $v_1$.
By \Cref{lem:earliest-bucket} applied to $\Pi_Q$, $H_L$ contains an edge $(v_1,x)$ to a portal $x$ such that $x$ is not later than $v_2$ on $Q$ and
\begin{equation*}
d_G(v_1,x)\le (1+\eps_0)d_G(v_1,v_2).
\end{equation*}
By construction of $F_Q^{\delta}$, there is an exact $x\leadsto v_2$ path in $G\cup F_Q^{\delta}$ using at most two hops, and
\begin{equation*}
d_G(x,v_2)\le w(Q)\le \eta L \le \eta(1+\eps_0)d_G(s,t).
\end{equation*}
Let $\widehat P$ be the path that follows $P[s,v_1]$, then the edge $(v_1,x)$, then the exact $x\leadsto v_2$ route inside $G\cup F_Q^{\delta}$, and finally the walk from $v_2$ to $t$.
Its hop count satisfies
\begin{equation*}
|\widehat P| \le |P[s,v_1]| + 1 + 2+10\beta = O(\beta).
\end{equation*}
For the weight, exactly as in \Cref{lem:route-from-sampled-suffix},
\begin{align*}
 w(\widehat P)
 &\le d_G(s,t) + \eps_0 d_G(v_1,v_2) + d_G(x,v_2)\\
 &\le \bigl(1+\eps_0+\eta(1+\eps_0)\bigr)d_G(s,t)
 \le (1+\eps)d_G(s,t).
\end{align*}
Thus $H_L$ supplies a $(1+\eps)$-approximate path with at most $O(\beta)$ hops for every pair at scale $L$.
A union bound over all scales and all ordered pairs shows that, with high probability, one random choice works simultaneously for every pair.
Therefore
\begin{equation*}
|H|=\Oish_\eps\!\left(\left(\frac{n}{\beta}\right)^{3/2}\right),
\end{equation*}
and $H$ is a $(1+\eps)$-hopset of hopbound at most $O(\beta)$ computed in time $\Oish_\eps(n/\beta\cdot m)$.
\end{proof}

\section*{AI Use Declaration}
An initial version of this paper was prepared entirely without the use of AI.
Afterwards, AI was used to implement reviewer feedback and minor fixes.

\printbibliography[title={References}]

\end{document}